\documentclass[10pt,conference]{IEEEtran}
\IEEEoverridecommandlockouts

\usepackage{cite}
\usepackage{amsmath,amssymb,amsfonts,amsthm}
\usepackage{relsize} % bigger summations
\usepackage{algorithm}
\usepackage{algpseudocode}
\usepackage{graphicx}
\usepackage{textcomp}
\usepackage{booktabs}
\usepackage{xspace}
\usepackage[dvipsnames]{xcolor}
\usepackage{subcaption}
\usepackage{listings}
\usepackage{gensymb}
\usepackage{wrapfig}
\usepackage{tabularx} % for better table formatting
\usepackage{booktabs} % for horizontal rules
\usepackage{makecell} % for vertical alignment
\usepackage{eso-pic}  % for the "do not distribute" text
\def\BibTeX{{\rm B\kern-.05em{\sc i\kern-.025em b}\kern-.08em
    T\kern-.1667em\lower.7ex\hbox{E}\kern-.125emX}}

\usepackage{tikz}
\tikzset{SmallBox/.style={draw,rectangle,minimum width=1.8cm,minimum height=.5cm,font=\scriptsize}}
\usetikzlibrary{
    decorations,
    patterns,
    decorations.pathmorphing,
    automata,
    positioning,
    backgrounds,
    fit,
    chains,
    scopes,
    shapes,
    matrix,
    shadows,
    arrows.meta,
    calc,
}

\makeatletter
\pgfdeclareshape{document}{
\inheritsavedanchors[from=rectangle] % this is nearly a rectangle
\inheritanchorborder[from=rectangle]
\inheritanchor[from=rectangle]{center}
\inheritanchor[from=rectangle]{north}
\inheritanchor[from=rectangle]{south}
\inheritanchor[from=rectangle]{west}
\inheritanchor[from=rectangle]{east}
\backgroundpath{% this is new
\southwest \pgf@xa=\pgf@x \pgf@ya=\pgf@y
\northeast \pgf@xb=\pgf@x \pgf@yb=\pgf@y
\pgf@xc=\pgf@xb \advance\pgf@xc by-10pt % this should be a parameter
\pgf@yc=\pgf@yb \advance\pgf@yc by-10pt
\pgfpathmoveto{\pgfpoint{\pgf@xa}{\pgf@ya}}
\pgfpathlineto{\pgfpoint{\pgf@xa}{\pgf@yb}}
\pgfpathlineto{\pgfpoint{\pgf@xc}{\pgf@yb}}
\pgfpathlineto{\pgfpoint{\pgf@xb}{\pgf@yc}}
\pgfpathlineto{\pgfpoint{\pgf@xb}{\pgf@ya}}
\pgfpathclose
\pgfpathmoveto{\pgfpoint{\pgf@xc}{\pgf@yb}}
\pgfpathlineto{\pgfpoint{\pgf@xc}{\pgf@yc}}
\pgfpathlineto{\pgfpoint{\pgf@xb}{\pgf@yc}}
\pgfpathlineto{\pgfpoint{\pgf@xc}{\pgf@yc}}
}
}
\makeatother

\newboolean{showcomments}
\setboolean{showcomments}{true} %set false for final submission
\makeatletter
\newcommand{\mynote}[3]{%
  \ifthenelse{\boolean{showcomments}}{%
   \fbox{\bfseries\sffamily\scriptsize#1}%
   {
   \small$\blacktriangleright$
   \textsf{\emph{\color{#3}{#2}}}
   $\blacktriangleleft$
   }
   }%
  {%
   \@bsphack
   \@esphack
  }%
}
\makeatother

\definecolor{asparagus}{rgb}{0.53, 0.66, 0.42}

\newcommand{\code}[1]{{\lstinline[basicstyle=\smaller\ttfamily]{#1}}}%
\newcommand{\opcode}[1]{{\lstinline[basicstyle=\smaller\ttfamily]{#1}}}%
\newcommand{\register}[1]{{\lstinline[basicstyle=\smaller\ttfamily]{#1}}}%
\newcommand{\reactor}[1]{{\lstinline[basicstyle=\smaller\sffamily]{#1}}}%
\newcommand{\port}[1]{{\lstinline[basicstyle=\smaller\sffamily]{#1}}}%
\newcommand{\us}[1]{$#1 \mu s$}%
\newcommand{\s}[1]{$#1 s$}%
\newcommand{\trigger}[1]{{\lstinline[basicstyle=\smaller\sffamily]{#1}}}%
\newcommand{\term}[1]{\textbf{#1}}%

\newcommand{\pretvm}[0]{\textsc{PretVM}\xspace}

\newcommand{\lflong}[0]{\textsc{Lingua Franca}\xspace}
\newcommand{\lf}[0]{\textsc{LF}\xspace}
\newcommand{\lb}[0]{\textsc{Load Balanced}\xspace}
\newcommand{\egs}[0]{\textsc{Edge Generation}\xspace}

\newtheorem{theorem}{Theorem}[section]

\newtheorem{lemma}[theorem]{Lemma}
\newtheorem{example}{Example}
\theoremstyle{definition}
\newtheorem{definition}{Definition}[section]

\newcommand{\eg}[0]{\textit{e.g.,}\xspace}
\newcommand{\ie}[0]{\textit{i.e.,}\xspace}
\newcommand{\wrt}[0]{\textit{w.r.t.}\xspace}

\begin{document}

% \AddToShipoutPicture{%
%     \AtPageUpperLeft{%
%         \raisebox{-\height}{%
%             \raisebox{-1cm}{% Move the text down by 2cm
%                 \makebox[\paperwidth]{%
%                     \begin{minipage}{\textwidth}
%                         \centering
%                         \textbf{THIS IS AN UNPUBLISHED DRAFT.\\PLEASE DO NOT DISTRIBUTE.}
%                     \end{minipage}
%                 }
%             }
%         }
%     }
% }

\title{\pretvm: Predictable, Efficient Virtual Machine for Real-Time Concurrency
\thanks{The work in this paper was supported in part by the National Science Foundation
(NSF), award \#CNS-2233769 (Consistency vs. Availability in Cyber-Physical
Systems), by DARPA grant FA8750-20-C-0156, by Intel, and by the iCyPhy Research
Center (Industrial Cyber-Physical Systems), supported by Denso, Siemens, and
Toyota. This work was also supported, in part, by the German Federal Ministry of
Education and Research (BMBF) as part of the program “Souverän. Digital.
Vernetzt.”, joint project 6G-life (16KISK001K), and by the German Research
Council (DFG) through the InterMCore project (505744711).Mirco Theile and Binqi
Sun were supported by the Chair for Cyber-Physical Systems in Production
Engineering at TUM.}}

%% Use this for anonymous submission.
% \author{\IEEEauthorblockN{Anonymous Authors}}

%% Enable line breaks in the author field.
\def\BibTeX{{\rm B\kern-.05em{\sc i\kern-.025em b}\kern-.08em
    T\kern-.1667em\lower.7ex\hbox{E}\kern-.125emX}}
    \makeatletter
\newcommand{\linebreakand}{%
  \end{@IEEEauthorhalign}
  \hfill\mbox{}\par
  \mbox{}\hfill\begin{@IEEEauthorhalign}
}
\makeatother

%% Comment out below for anonymous submission.
\author{\IEEEauthorblockN{Shaokai Lin}
\IEEEauthorblockA{
\textit{UC Berkeley}\\
Berkeley, USA}
\and
\IEEEauthorblockN{Erling Jellum}
\IEEEauthorblockA{
\textit{NTNU}\\
Trondheim, Norway}
\and
\IEEEauthorblockN{Mirco Theile}
\IEEEauthorblockA{
\textit{TU Munich}\\
Munich, Germany}
\and
\IEEEauthorblockN{Tassilo Tanneberger}
\IEEEauthorblockA{
\textit{TU Dresden}\\
Dresden, Germany}
\and
\IEEEauthorblockN{Binqi Sun}
\IEEEauthorblockA{
\textit{TU Munich}\\
Munich, Germany}
\linebreakand
\IEEEauthorblockN{Chadlia Jerad}
\IEEEauthorblockA{
\textit{University of Manouba}\\
Manouba, Tunisia}
\and
\IEEEauthorblockN{Ruomu Xu}
\IEEEauthorblockA{
\textit{UC Berkeley}\\
Berkeley, USA}
\and
\IEEEauthorblockN{Guangyu Feng}
\IEEEauthorblockA{
\textit{UC Berkeley}\\
Berkeley, USA}
\and
\IEEEauthorblockN{Christian Menard}
\IEEEauthorblockA{
\textit{TU Dresden}\\
Dresden, Germany}
\linebreakand
\IEEEauthorblockN{Marten Lohstroh}
\IEEEauthorblockA{
\textit{UC Berkeley}\\
Berkeley, USA}
\and
\IEEEauthorblockN{Jeronimo Castrillon}
\IEEEauthorblockA{
\textit{TU Dresden}\\
Dresden, Germany}
\and
\IEEEauthorblockN{Sanjit Seshia}
\IEEEauthorblockA{
\textit{UC Berkeley}\\
Berkeley, USA}
\and
\IEEEauthorblockN{Edward Lee}
\IEEEauthorblockA{
\textit{UC Berkeley}\\
Berkeley, USA}
}

\maketitle

\begin{abstract}
This paper introduces the Precision-Timed Virtual Machine (\pretvm), an
intermediate platform facilitating the execution of quasi-static schedules
compiled from a subset of programs written in the Lingua Franca (\lf)
coordination language.  The subset consists of those programs that in principle
should have statically verifiable and predictable timing behavior.  The \pretvm
provides a schedule with well-defined worst-case timing bounds. 
The \pretvm provides a clean separation between \emph{application} logic and
\emph{coordination} logic, yielding more analyzable program executions.
Experiments
compare the \pretvm against the default (more dynamic) LF scheduler and show that
it delivers time-accurate deterministic execution. 
\end{abstract}


\section{Introduction}
When developing real-time software for Cyber-Physical Systems (CPSs), the conventional strategy often adopts a
bottom-up approach, where the emphasis is first placed on defining the "how"—for
instance, using a preemptive EDF scheduler—in hopes of realizing the "what",
i.e., the desired timing behavior.  
This bottom-up process relies heavily on
experimental feedback to find out whether a system is able to meet system-level
timing constraints. Iteratively, code optimizations are carried out until the
desired behavior is achieved. 
Yet, the bottom-up method complicates code
validation and upgrades, and it renders implementations platform-specific
and only works well in systems 
where multiple runs of the same code yield very similar timing behavior. 
Code obtained through such 
a process, therefore, tends to not be useable on different platforms and is
challenging to prove safe.
As Cyber-Physical Systems become more software-defined (e.g., in the automotive
industry~\cite{liu2022impact}) and increasingly often implemented using
multi-core and heterogeneous hardware, the bottom-up approach becomes less
feasible and attractive.

Henzinger and Kirsch~\cite{Henzinger:03:Giotto} observe that a lot of
these difficulties are due to an entanglement of concerns between
``reactivity,'' which refers to real-time interaction with a physical
environment, ``schedulability,'' which concerns real-time execution in a
specific execution environment, and ``functionality,'' such as control laws and
device drivers. The Giotto language and methodology they developed, which
introduced the concept of Logical Execution Time (LET), meant specifically
to help programmers separate these concerns. A Giotto program explicitly
specifies the exact real-time interaction of software components with the
physical world, leaving it to the Giotto compiler (not the software or control
engineer) to generate timing code that ensures the specified timing behavior on a given
platform. The timing code, known as E code, encodes Giotto's LET semantics and
runs on the Embedded Machine, a virtual machine that makes E code portable~\cite{Henzinger:02:EMachine}.

The recently emerged Lingua Franca (LF)~\cite{LohstrohEtAl:21:Towards,lohstroh:tecs23}
coordination language and its underlying reactor model~\cite{Lohstroh:2019:CyPhy} can be characterized as a
generalization of Giotto and
LET~\cite{LeeLohstroh:22:LET}. It takes the separation between reactivity and
schedulability a step further by decoupling the passage of physical time from
logical time, except when interactions with the physical world are involved
(through timers, sensors, and actuators)~\cite{LohstrohEtAl:20:LF}.
In part because LF is not strictly
task-based and allows sporadic events, the scheduling of LF programs is carried
out dynamically, at runtime. While the timing behavior is specified in the LF
code, alignment with physical time is best-effort and without hard guarantees
that could be derived based on a quasi-static schedule.

This paper describes an augmentation of the LF toolchain that generates
quasi-static schedules.
We do this through a compiler extension, runtime integration, and virtual machine abstraction called the \term{Precision-Timed Virtual Machine},
or in short, \pretvm. The functionality that \pretvm
offers to LF is similar to what the E Machine~\cite{Henzinger:02:EMachine}
provides for LET. \pretvm enables a methodology for obtaining predetermined
temporal behavior for the subset of LF programs comprising sequences of
periodic task sets.
Based on a directed acyclic graph (DAG) representation of the program, a 
\term{partitioned quasi-static schedule} can be derived using different schedulers. 
This quasi-static schedule, after being translated into our \pretvm bytecode for each partition, can be executed using multiple worker threads. 
We show that our approach exhibits predictable timing and can leverage existing work in multi-core mapping and
scheduling, such as Edge Generation Scheduling for DAG Tasks~\cite{sun2023edge}.

In this work, we focus on a subset of \lf programs whose behavior can be
statically analyzed by our approach. We exclude ``dynamic'' features of \lf, such
as actions and physical connections, and save them for future work.

\paragraph*{Contributions}
In this paper, we present the following key contributions:
\begin{enumerate}
    \item We present the Precision-Timed Virtual Machine, \pretvm, as an intermediary
    platform supporting a top-down, model-based design flow for real-time
    software in Cyber-Physical Systems.
    \item We provide a methodology to statically schedule and compile a subset
    of \lf into \pretvm bytecode.
    \item We show that our approach is amenable to predicting
    the worst-case execution time (WCET) of systems' hyperperiods.
    \item We evaluate our approach using a set of \lf benchmarks evaluating our methodology \wrt
    timing accuracy. Results show that
    \pretvm delivers time-accurate deterministic execution.
\end{enumerate}

\section{Running Example}\label{sec:running_example}
\subsection{Reaction Wheel System}
\begin{figure}[h!]
  \centering
  \includegraphics[width=\columnwidth]{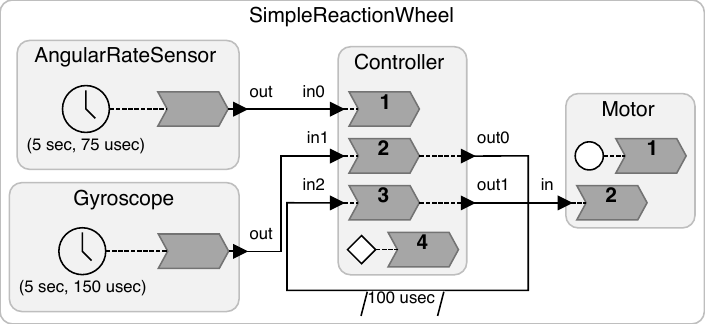}
  \caption{A simple reaction wheel system}
  \label{fig:example}
\end{figure}

To informally introduce the reactor model, we present a simplified example
of a reaction wheel system used for satellite attitude control~\cite{cardoso2022reactionwheel}. In this
system, a gyroscope and an angular rate sensor periodically collect data from the environment and
send them to a controller, which performs sensor fusion and generates a command
the motor.

The diagram in Figure~\ref{fig:example} represents the
simplified reaction wheel system. A \term{reactor} is a stateful concurrent component.
In the diagram, reactors are represented as rectangles with rounded corners, and
they are \reactor{Gyroscope}, \reactor{AngularRateSensor}, \reactor{Controller},
and \reactor{Motor}. A reactor can
contain \term{reactions}, \term{timers}, \term{ports},
and \term{connections}.

A reaction is a routine that executes when any of its triggers is present, and
it is rendered as a chevron in the diagram. The body of a reaction is written in
a \term{target language} and represents the ``business logic'' of the
application. 
In this work, each reaction is further annotated with its worst-case execution
time (WCET) for a given platform.
The \reactor{Controller} reactor has four reactions, each labelled with an integer number. These numbers correspond to the \term{priorities} of the
reactions within \reactor{Controller}, which determine the order in which reactions
execute should they be triggered simultaneously.
It is important to note that reaction
priorities establish a total order only \emph{within} the same reactor.

A reaction can be triggered by a port, a timer, an action, or a built-in trigger
such as \trigger{startup} or \trigger{shutdown}. A port, rendered as a black
triangle in the diagram, is used to communicate with other reactors. A timer,
rendered as a clock in the diagram, is used to generate periodic events. Timers
are characterized by an offset and a period which are shown below the clock. For
example, the timer in \reactor{AngularRateSensor} starts firing at \s{5} with 
a period of \s{75}.

Around the \reactor{Controller} reactor, a connection is drawn from its output port \port{out0} to its input port \port{in2},
with a delay of \us{100}, effectively scheduling an  event within itself
\us{100} into the future. The delay
annotated to connections is optional. If no delay is specified, the event
travels along a connection logically instantaneously.

The built-in trigger \trigger{startup}, rendered as a white circle, is present at
the beginning of the execution, and the built-in trigger \trigger{shutdown},
rendered as a white diamond, is present at the last tag of the execution.
Besides triggers, connections can be drawn between two ports using
connection statements.

All events in the system are handled in the order of \term{tags}, which use a
super-dense representation of time~\cite{manna1993verifying}, encoded as
tuples of the form $g=(t, m)$, where $t \in \mathbb{T}$ is the \term{time value}
and $m \in \mathbb{N}$ is a \term{microstep index}. Reactions are logically
instantaneous; \term{logical time} does not elapse during the execution of a
reaction (\term{physical time}, denoted as $T$, on the other hand, does elapse). If a reaction
produces an output that triggers another reaction, then the two reactions execute
logically simultaneously (\ie at the same tag).

\section{Background}

\subsection{Reactor Model and Lingua Franca}
The reactor model~\cite{Lohstroh:2019:CyPhy} is a deterministic
model of concurrent computation for reactive systems. It is an adaptation of the
actor model~\cite{DBLP:conf/ijcai/HewittBS73} in which messages flow through
ports, and delivery is not address-based but connection-based, like in
dataflow~\cite{Dennis:74:Dataflow,NajjarEtAl:99:Dataflow} and Kahn process
networks~\cite{Kahn:74:PN}. Messages are tagged with a timestamp and delivered
in order, assigning a discrete-event semantics to the operation of reactors.
Unlike LET~\cite{Henzinger:01:Giotto}, in which logical execution time is
synonymous with physical time, reactions to events are logically instantaneous,
following the synchronous hypothesis that is central to the synchronous
languages~\cite{Benveniste:91:Synchronous}. Still, a LET can be assigned to a
reaction, simply by adding a logical delay to its output. The execution
semantics of reactors ensure that at any tag, no reaction executes before all
events that it depends on are known. This invariant is independent of the
physical time that elapses during the execution of a reactor program, and it
ensures that reactors behave deterministically. The proximity between the tag at
which an event is scheduled and the physical time at which it is reacted to,
\emph{does} depend on execution times, and requires further analysis in order to
provide any assurances in that regard.

LF is a polyglot coordination language that allows for the definition and
composition of reactors. Event-triggered reactions, which constitute the
functionality of reactors, have an interface specified in LF syntax, but an
implementation specified in plain target code. Currently, LF supports C, C++,
Python, Rust, and TypeScript. LF code gets compiled down to target code that is
executable using a runtime environment that coordinates the scheduling of events
and reactions to events. For most targets, the runtime transparently exploits
available parallelism and utilizes multiple cores. LF has support for multiple
platforms, including POSIX, Arduino, Zephyr, and several bare-iron
microcontrollers. LF programs can also be federated, meaning that the reactor
semantics get preserved in a distributed system, across separate processes that
communicate over a network. Such federated execution is outside the scope of
this paper.

\subsection{Directed Acyclic Graphs (DAGs)}
DAGs are widely used in real-time applications, such as automotive and avionics, to model real-time computing tasks and their precedence constraints. This subsection introduces basic concepts and notations for real-time DAG task modeling. 
\subsubsection{Task Model}
A DAG task is characterized by $(\mathcal{G}, P)$, in which $\mathcal{G}$ is a graph defining the set of sub-tasks and $P$ denotes the DAG task period. The graph $\mathcal{G}$ comprises of $(\mathcal{V},\mathcal{E})$, where $\mathcal{V} = \{v_i\}$ is a set of $n$ nodes representing $n$ sub-tasks, and $\mathcal{E} = \{e_{ij}\}$ is a set of directed edges representing the precedence relation between the sub-tasks. 
For any two nodes $v_i$ and $v_j$ connected by a directed edge $e_{ij}$, $v_j$ can start execution only if $v_i$ has finished its execution. Given the edge $e_{ij}$, $v_i$ is the predecessor of $v_j$. Each sub-task $v_i$ is a non-preemptable sequential computing workload, with its worst-case execution time (WCET) denoted as $C_i$.

\subsubsection{Node-level Timing Attributes}
The timing constraints of each node can be defined through four attributes: earliest starting time (EST), earliest finishing time (EFT), latest starting time (LST), and latest finishing time (LFT). The EST is the earliest time a node can start executing, equalling the maximum of its predecessors' EFTs. Similarly, the LFT defines the latest time a node is allowed to finish its execution to meet the deadline, which is equal to the minimum of its successors' LST. The relationship between the timing attributes is
\begin{equation}
    \begin{aligned}
        t^\text{EST}_i &= \max\{t^\text{EFT}_j ~\vert~ e_{ji}\in\mathcal{E}\}\\
        t^\text{LFT}_i &= \min\{t^\text{LST}_j ~\vert~ e_{ij}\in\mathcal{E}\}
    \end{aligned}
\end{equation}
and $t^\text{EFT}_i = t^\text{EST}_i + C_i$ and $t^\text{LST}_i = t^\text{LFT}_i - C_i$. The
 $t^\text{LFT}$ of all nodes without a successor is set to the DAG task period~$P$.

\subsection{DAG Modeling for LF Programs}
When constructing the DAG task for DAG scheduling it is necessary that all
externally defined constraints are represented through the DAG topology and the
scalar node values. A real-time \lf program releases a task $v_i$ (\ie a
reaction invocation) at a time offset $O_i$ and expects it to complete by
a deadline $D_i$, where $D_i$ could be the release time of the reaction's next 
invocation or the end of the hyperperiod. The DAG needs to be adapted to reflect
these constraints, which are translated to constraints on node-level
timing attributes as  
\begin{equation}\label{eq:timing_req}
    O_i \leq t^\text{EST}_i  \quad \text{and} \quad    t^\text{LFT}_i \leq D_i,
\end{equation}
i.e., the earliest starting time needs to be larger than the offset and the last
finishing time needs to be less than the deadline. To constrain the timing
attributes accordingly, we borrow the concept of \term{virtual nodes}
from~\cite{verucchi2020latency}. There are two types of virtual nodes,
\term{dummy nodes} $d_l$ that have a predefined execution time and
\term{sync nodes} $s_l$ with zero execution time. Using the virtual nodes, a
\term{virtual path} can be constructed, alternating $S$ sync and $S-1$ dummy
nodes. 
The total execution time of the virtual path is equal to the period of
the DAG task. Through the defined execution time of the dummy nodes, it follows
that for each sync node $s_l$, 
\begin{equation}
    t^\text{EST}_{s_l} = t^\text{EFT}_{s_l} = t^\text{LST}_{s_l} = t^\text{LFT}_{s_l} = \sum_{i=1}^{l-1} C_{d_i} = P - \sum_{j=l}^{S-1} C_{d_j}.
\end{equation}
The nodes in the virtual path can be designed such that all unique offset and
deadline values are represented by a sync node. Adding an edge from the
appropriate sync node $s_l$ to a sub-task $v_k$, or vice-versa, respectively,
results in the relation  
\begin{equation}
    O_k = t^\text{EFT}_{s_l} \leq t^\text{EST}_k   \quad \text{or} \quad   t^\text{LFT}_{k} \leq t^\text{LST}_{s_l} = D_k,
\end{equation}
which enforces \eqref{eq:timing_req} within the DAG task.
As a concrete example, Figure~\ref{fig:dags_gen} shows the DAG encoding of the
running example.

\definecolor{external}{RGB}{223, 203, 153}
\definecolor{frontEnd}{RGB}{191, 213, 233}
\definecolor{backEnd}{RGB}{196, 226, 185}

\tikzstyle{doc}=[%
draw,
thin,
color=black,
shape=document,
minimum width=1cm,
minimum height=1.5cm,
inner sep=2ex,
]

\tikzstyle{tool}=[
SmallBox,
thick,
rounded corners=5pt,
minimum width=8.6cm,
minimum height=7cm,
text width=1.2cm,
align=center
]

\begin{figure}[t]
  \begin{tikzpicture}[node distance=1cm,
    prior/.style={rectangle, draw, text centered, minimum height=1cm}, 
    box/.style={rectangle, draw, text centered, minimum height=1cm},
    innerbox/.style={rectangle, draw, inner sep=1em},
    decision/.style={diamond, draw, text badly centered, inner sep=-2pt, aspect=1.3},
    arrow/.style={thick,-{Latex[scale=1.2]}}
]
    % Nodes
    \node [tool] (compiler) {};
    \node[font=\normalsize, above=0.1cm of compiler.south, xshift=-1cm] (l)  {\lflong \textsc{Compiler}};

    \node [prior, below=0.4cm of compiler.north, xshift=-3.6cm, align=center] (parser) {LF\\Parser};
    \node [prior, right=of parser, align=center] (ast) {AST\\Builder};
    \node [prior, right=of ast, align=center] (cgen) {C\\Generator};
    \node [box, fill=frontEnd, below=1.8cm of parser.west, anchor=west, align=center] (daggen) {DAG\\Generator};
    \node [box, fill=frontEnd, below=1.8cm of cgen.east, anchor=east, align=center] (explorer) {State Space\\Explorer};
    \node [box, fill=backEnd, right=0.5cm of cgen, align=center] (bcgen) {Bytecode\\Generator};
    
    \node [box, fill=frontEnd, below=3.3cm of bcgen.east, anchor=east, minimum width=2.5cm, minimum height=1cm, align=center] (lb) {\textsc{Load}\\\textsc{Balanced}};
    \node [box, dashed, fill=external, below=0.1cm of lb, minimum width=2.5cm, minimum height=1cm, align=center] (egs) {\textsc{Edge}\\\textsc{Generation}};
    % \node [box, dashed, fill=external, below=0.1cm of egs, minimum width=2.5cm, minimum height=1cm, align=center] (mbox) {\textsc{Mocasin}};
    % \node [innerbox, minimum width=2.7cm, minimum height=3.45cm, above=0.1cm of
    % lb, anchor=north, align=center, label=below:{Static Schedulers}]
    % (schedulers) {};
    \node [innerbox, minimum width=2.7cm, minimum height=2.35cm, above=0.1cm of
    lb, anchor=north, align=center, label={[align=center]below:{Quasi-Static\\Schedulers}}] (schedulers) {};
    
    \node [decision, fill=frontEnd, left=3.5cm of egs.west, align=center, text depth=1ex] (decision) {\textbf{Select}\\\textbf{Scheduler}};

    % \node [box, fill=frontEnd, left=1.5cm of mbox, align=center] (sdfggen) {SDFG\\Generator};

    \node (lf) [doc, above=2.25cm of parser.west, anchor=west, align=center] {\textsc{LF} w/\\WCETs};
    \node (inst) [doc, above=of cgen, align=center] {Inst.\\Code};
    \node (lib) [doc, left=0.5cm of inst, align=center, fill=backEnd] {\pretvm\\Runtime};
    \node (bc) [doc, fill=backEnd, above=of bcgen, align=center] {\pretvm\\Bytecode};
    \node [innerbox, dashed, minimum width=6.6cm, minimum height=1.8cm, left=0.1cm of lib, anchor=west, align=center, label=above:{Compiled by \textsc{GCC} into executable}] (output) {};

    % Lines
    \draw [arrow] ([xshift=-0.33cm]lf.south) -- (parser.north);
    \draw [arrow] (parser) -- (ast);
    \draw [arrow] (ast) -- (cgen);
    \draw [arrow] ([xshift=0.25cm]daggen.south) -- node[midway, above, rotate=90] {DAGs} (decision.north);

    % Define a coordinate for the split point
    \coordinate[right=1.17cm of decision.east] (split);

    % Draw the main line from decision to the split point
    \draw[line width=0.8pt] (decision.east) -- (split);

    % Draw the individual arrows from the split point to their destinations
    % \draw [arrow] (split) -| node[near end, right] {\textbf{M}} (sdfggen.north);
    \draw [arrow] (split) -- node[above] {\textbf{EG}} (egs);
    \draw [arrow] (split) |- node[near end, above] {\textbf{LB}} (lb);

    \draw [arrow] (ast.south) |- ([yshift=-0.3cm]ast.south) -| (explorer.north);
    \draw [arrow] (explorer) -- (daggen) node[midway, below, align=center, yshift=0.5cm] {SSDs w/\\Guarded\\Transitions};
    % \draw [arrow] (sdfggen.east) -- node[above] {SDFGs} (mbox.west);
    \draw [arrow] ([xshift=0.43cm]schedulers.north) -- node[align=center, rotate=90] {Static\\Schedules} (bcgen);
    \draw [arrow] (cgen) -- (inst);
    \draw [arrow] (bcgen) -- (bc);

  \end{tikzpicture}

\caption{Methodology overview. Colored boxes are our work. Brown boxes denote
external tools we build interface for. Blue and brown boxes are
described in Sec.~\ref{sec:lf_to_static_schedules}. Green boxes are
described in Sec.~\ref{sec:static_schedules_to_pretvm}.}
\label{fig:pretvm_methodology}
\end{figure}
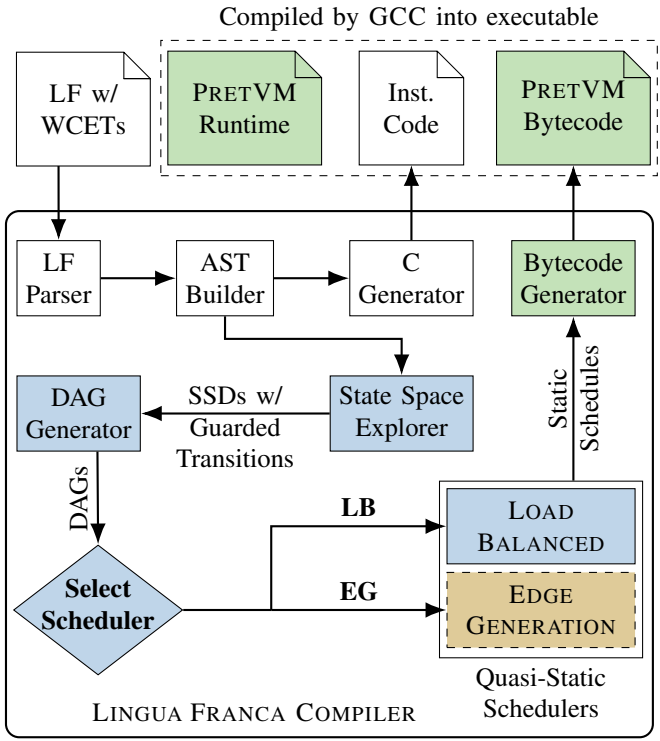

\section{From LF to Quasi-Static Schedules}
\label{sec:lf_to_static_schedules}

Let us now dive into our methodology, as shown in Figure~\ref{fig:pretvm_methodology}.
At a high level, our methodology consists of two parts: first,
compiling an \lf program into quasi-static schedules
(Sec.~\ref{sec:lf_to_static_schedules}), and second, generating and executing 
bytecode on \pretvm implementations (Sec.~\ref{sec:static_schedules_to_pretvm}).
In this section, we focus on the first part, which covers the blue and
brown boxes in Figure~\ref{fig:pretvm_methodology}.

\subsection{State Space Diagrams with Guarded Transitions}
The input of our approach is an \lf program with
WCET annotations. Each reaction is labeled by the user with a WCET annotation,
indicating the maximum amount of time this reaction is expected to require to complete execution
on a certain execution platform along with the associated overhead. We explain
further in Sec.~\ref{sec:predictable_timing} how a user could derive the WCETs. Given an annotated \lf program, the \lf
compiler parses the program and builds an abstract syntax tree (AST).
Once an AST is built, the compiler invokes a built-in C Generator for generating
program-specific instrumentation code, including custom C structs for
reactor definitions in the specific program, user-specified reaction bodies
written in C, memory allocation and deallocation functions, etc.

From the AST, the \term{State Space Explorer} generates \term{State Space
 Diagrams} (SSDs)~\cite{lin2023lfverifier}, which capture the logical behavior
 of an \lf program.
 Lin et al.~\cite{lin2023lfverifier} identified a subset of \lf programs whose
 logical behavior can be represented using a single SSD.
 In this work,
 we expand this subset of \lf programs by combining multiple SSDs with guarded transition
 relations between each pair of SSDs.

An SSD is a directed graph $(V, E)$, where each node consists of a timestamp,
a set of reactions invoked at this timestamp, and a set of future events known
at this timestamp.
Formally, $V = (t, r, e) \subseteq
\textit{tags} \times \mathcal{P}(\textit{rxns}) \times
\mathcal{P}(\textit{events})$ and $E \subseteq V \times V$, where $\mathcal{P}$
denotes a power set, $\textit{tags}$ denotes the set of timestamps,
$\textit{rxns}$ denotes the set of reactions in an \lf program, and
$\textit{events}$ denotes the set of events.
To construct an SSD, the State Space Explorer performs a light-weight simulation
by unrolling a worst-case execution until either 
\begin{enumerate}
  \item there are no more pending events, \ie $e = \varnothing$,
  \item or the simulation completes a user-specified time horizon $h$, \ie $t = h$,
  \item or a hyperperiod (\ie a cycle) is found.
\end{enumerate}
The simulation approach is similar
to~\cite{ghamarian2006throughput}.

\begin{figure}
  \centering
  \includegraphics[width=0.7\columnwidth]{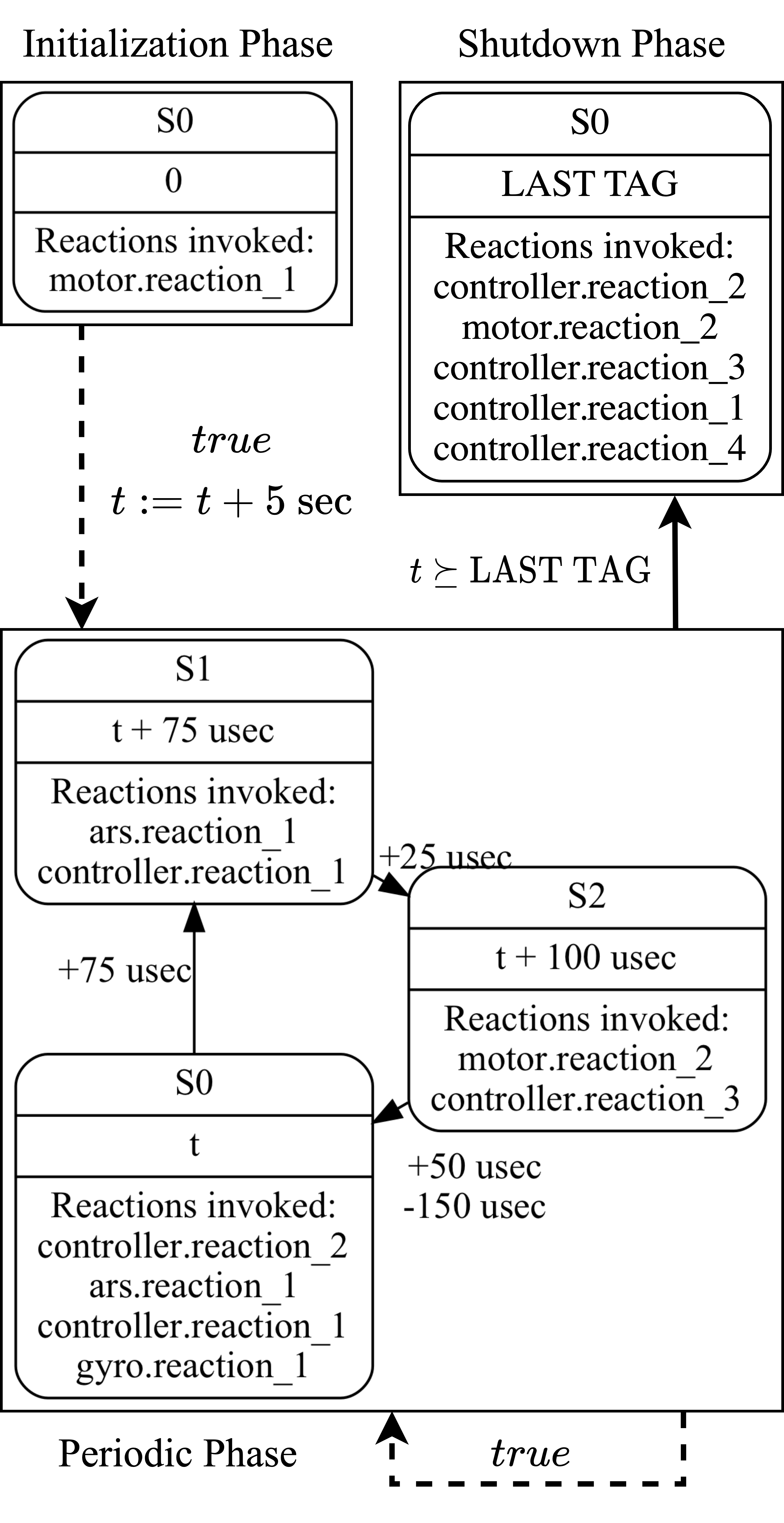}
  \caption{Three separate SSDs connected by guarded transitions from the running example. Pending events are not rendered here to save space.}
  \label{fig:state_space_diagrams}
\end{figure}

In this work, we combine multiple SSDs to represent more dynamic behaviors of \lf programs.
Figure~\ref{fig:state_space_diagrams} shows the SSDs for the running example in
Sec.~\ref{sec:running_example}. The logical behavior of the reaction
wheel system can be divided into three phases: an \term{initialization phase}, a
\term{periodic phase}, and a \term{shutdown phase}. 
Each phase can be represented by an individual SSD, respectively.
The initialization phase
consists of one state space node with timestamp $0$, the very first timestamp of an
execution, and a single reaction invocation, reaction 1 from the
\reactor{Motor} reactor, due to the startup trigger (rendered as a white
circle in Figure~\ref{fig:example}). Similarly, the shutdown phase of the
execution also consists of a single state space node with the
last tag of the execution---a user-specified timeout---and five possible reaction invocations due to
the shutdown trigger and the last firings of the timers.
The periodic phase contains three state space nodes: the first node starts at
\s{t=5}, the second starts \us{75} later, and the third node
starts \us{25} after the second. The reaction invocations occur due to
the timer firings and the production of outputs from the reactions in the
\reactor{AngularRateSensor} and \reactor{Gyroscope}.

The three phases of execution are connected using guarded transitions. Here, we
use the arrow notations from~\cite{LeeSeshia:17:EmbeddedSystems}, where the
solid arrows represent regular guarded transitions, and
dashed arrows represent \textit{default} transitions, which are taken when none
of the guarded transitions are taken. In Figure~\ref{fig:state_space_diagrams},
the initialization phase has a default transition to the periodic phase, meaning
that when the task set in the initialization phase is done, the execution can
move onto the task set in the periodic phase. As part of the transition, the
current time $t$ is incremented by $5\text{sec}$ due to the five-second timer offset
specified in the program (Figure~\ref{fig:example}). The periodic phase has two
outgoing transitions. The first is a guarded transition to the shutdown phase,
with a guard $t \succeq \text{LAST TAG}$, where $t$ denotes the current time.
This transition is taken when the current time $t$ reaches $\text{LAST TAG}$,
which is a user-specified timeout value. While $\text{LAST TAG}$
is not reached, the execution does not take this guarded
transition, but instead takes the default transition of the periodic phase and
executes the periodic task set for another hyperperiod.

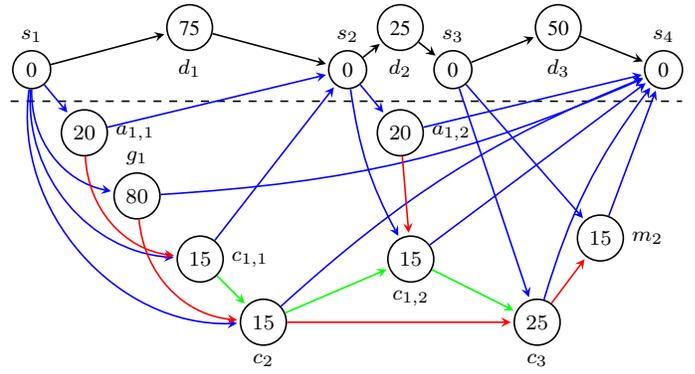
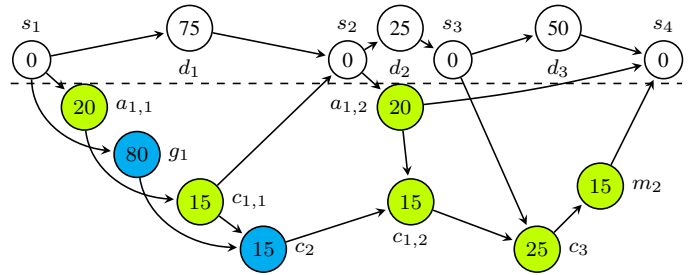
\begin{figure}
  \centering
  \begin{subfigure}{\columnwidth}
    \centering  
    \begin{tikzpicture}
        [
            ->,>=stealth,shorten >=1pt,auto,node distance=2.8cm,semithick,node font=\footnotesize, scale=0.7, xscale=0.8, yscale=0.8
        ]
        \node (s1) [circle, draw, label=above:$s_1$]  at (0, 1) {0};
        \node (d1) [circle, draw, label=below:$d_1$]  at (3.75,2) {75};
        \node (s2) [circle, draw, label=above:$s_2$]  at (7.5, 1) {0};
        \node (d2) [circle, draw, label=below:$d_2$]  at (8.75, 2) {25};
        \node (s3) [circle, draw, label=above:$s_3$]  at (10, 1) {0};
        \node (d3) [circle, draw, label=below:$d_3$]  at (12.5,2) {50};
        \node (s4) [circle, draw, label=above:$s_4$]  at (15, 1) {0};
        
        \draw (s1) -> (d1);
        \draw (d1) -> (s2);
        \draw (s2) -> (d2);
        \draw (d2) -> (s3);
        \draw (s3) -> (d3);
        \draw (d3) -> (s4);
        
        \draw[dashed, -] (-0.5, 0.25) -- (15.5, 0.25);

        \node (a11) [circle, draw, label=right:$a_{1,1}$]  at (1.25,-0.5) {$20$};
        \node (g1) [circle, draw, label=above:$g_{1}$]  at (2.5,-2) {$80$};
        \node (c11) [circle, draw, label=right:$c_{1,1}$]  at (4,-3.5) {$15$};
        \node (c2) [circle, draw, label=below:$c_{2}$]  at (5.5,-5) {$15$};

        \node (a12) [circle, draw, label=right:$a_{1,2}$]  at (8.75,-0.5) {$20$};
        \node (c12) [circle, draw, label=below:$c_{1,2}$]  at (9,-3.5) {$15$};

        \node (m2) [circle, draw, label=right:$m_{2}$]  at (13.5,-3) {$15$};
        \node (c3) [circle, draw, label=below:$c_{3}$]  at (12,-5) {$25$};

        \draw[color=blue] (s1) -> (a11);
        \draw[color=blue] (s1) edge[bend right=40] (g1);
        \draw[color=blue] (s1) edge[bend right=45] (c11);
        \draw[color=blue] (s1) edge[bend right=50] (c2);
        
        \draw[color=blue] (s2) -> (a12);
        \draw[color=blue] (s2) edge[bend right=10] (c12);

        \draw[color=blue] (s3) -> (c3);
        \draw[color=blue] (s3) -> (m2);

        \draw[color=blue] (a11) -> (s2);
        \draw[color=blue] (g1) edge[bend right=10] (s4);
        \draw[color=blue] (a12) -> (s4);
        \draw[color=blue] (m2) -> (s4);
        \draw[color=blue] (c11) -> (s2);
        \draw[color=blue] (c12) -> (s4);
        \draw[color=blue] (c2) edge[bend left=10] (s4);
        \draw[color=blue] (c3) edge[bend left=10] (s4);

        \draw[color=red] (a11) edge[bend right=40] (c11);
        \draw[color=red] (a12) -> (c12);
        \draw[color=red] (g1) edge[bend right=40] (c2);
        \draw[color=red] (c2) -> (c3);
        \draw[color=red] (c3) -> (m2);
        
        \draw[color=green] (c11) -> (c2);
        \draw[color=green] (c12) -> (c3);
        \draw[color=green] (c2) -> (c12);
    
    \end{tikzpicture}
    \vspace{-10pt}
    \caption{DAG representation after conversion from the state space diagram in Figure~\ref{fig:state_space_diagrams}, showing timing edges in blue, trigger edges in red, and reactor sequentialization edges in green, with the virtual path above the dashed line.}
    \label{fig:dag_gen}
  \end{subfigure}    
  \begin{subfigure}{\columnwidth}
  \centering
  \begin{tikzpicture}
        [
            ->,>=stealth,shorten >=1pt,auto,node distance=2.8cm,semithick,node font=\footnotesize, scale=0.7, xscale=0.8,, yscale=0.6
        ]
        \node (s1) [circle, draw, label=above:$s_1$]  at (0, 1) {0};
        \node (d1) [circle, draw, label=below:$d_1$]  at (3.75,2) {75};
        \node (s2) [circle, draw, label=above:$s_2$]  at (7.5, 1) {0};
        \node (d2) [circle, draw, label=below:$d_2$]  at (8.75, 2) {25};
        \node (s3) [circle, draw, label=above:$s_3$]  at (10, 1) {0};
        \node (d3) [circle, draw, label=below:$d_3$]  at (12.5,2) {50};
        \node (s4) [circle, draw, label=above:$s_4$]  at (15, 1) {0};
        
        \draw (s1) -> (d1);
        \draw (d1) -> (s2);
        \draw (s2) -> (d2);
        \draw (d2) -> (s3);
        \draw (s3) -> (d3);
        \draw (d3) -> (s4);
        
        \draw[dashed, -] (-0.5, 0.25) -- (15.5, 0.25);

        \node (a11) [circle, draw, label=right:$a_{1,1}$, fill=lime]  at (1.25,-0.5) {$20$};
        \node (g1) [circle, draw, label=right:$g_{1}$, fill=cyan]  at (2.5,-2) {$80$};
        \node (c11) [circle, draw, label=right:$c_{1,1}$, fill=lime]  at (4,-3.5) {$15$};
        \node (c2) [circle, draw, label=right:$c_{2}$, fill=cyan]  at (5.5,-5) {$15$};

        \node (a12) [circle, draw, label=left:$a_{1,2}$, fill=lime]  at (8.75,-0.5) {$20$};
        \node (c12) [circle, draw, label=below:$c_{1,2}$, fill=lime]  at (9,-3.5) {$15$};

        \node (m2) [circle, draw, label=right:$m_{2}$, fill=lime]  at (13.5,-3) {$15$};
        \node (c3) [circle, draw, label=right:$c_{3}$, fill=lime]  at (12,-5) {$25$};

        \draw (s1) -> (a11);
        \draw (s1) edge[bend right=40] (g1);
        
        \draw (s2) -> (a12);

        \draw (s3) -> (c3);

        \draw (c11) -> (s2);
        % \draw (g1) edge[bend left=5] (s5);
        \draw (a12) edge[bend right=5] (s4);
        \draw (m2) -> (s4);

        \draw (a11) edge[bend right=40] (c11);
        \draw (a12) -> (c12);
        \draw (g1) edge[bend right=40] (c2);
        \draw (c3) -> (m2);
        
        \draw (c11) -> (c2);
        \draw (c12) -> (c3);
        \draw (c2) -> (c12);

        % \draw[color=green] (c11) -> (a12);
    
    \end{tikzpicture}
    \vspace{-10pt}
  \caption{Sample output from the \textsc{Edge Generation} scheduler, which is a transitive reduction of the DAG in Figure~\ref{fig:dag_gen} with width $M=2+1$, with $+1$ for the virtual path. The colors of the nodes show a valid partitioning to 2 available cores.}
  \label{fig:dag_gen_red}
  \end{subfigure}
  \caption{DAG generation and partitioned scheduling. The node names are the
  initial of the reactor name with the reaction number and optionally the
  invocation count as subscript. For example, $a_{1,1}$ is the first invocation
  of reaction 1 in \reactor{AngularRateSensor}. The numbers in the
  nodes indicated the WCET of each reaction and the nodes above the dashed line
  are the virtual path.}
  \label{fig:dags_gen}
\end{figure}

\subsection{Converting SSDs to DAGs}
In Figure~\ref{fig:pretvm_methodology}, after the SSDs with guarded transitions
are generated by the State Space Explorer, the SSDs enter the DAG Generator and
are converted to DAGs. The conversion starts by generating nodes corresponding
to the reaction invocations, and then connects these nodes through edges based on
timing and precedence constraints. 

When iterating over the state space diagram, for each reaction invocation a node
is generated, taking the current logical time as offset, and adopting the
execution time of the reaction. In this work, we assume implicit deadlines for all nodes, meaning their deadline is equal to their period. As such, every node has a deadline according to the logical tag of its next invocation. The period of the DAG task is extracted from the length of the loop in
the state space diagram. After generating all nodes, the virtual nodes are
generated according to the unique offset and deadline values.

The edges added are the following:
\begin{enumerate}
    \item Edges connecting the virtual nodes to form the virtual path;
    \item Timing edges from and to the virtual path for offset and deadline values of the nodes, respectively;
    \item Dependency edges between nodes corresponding to reactions where the second is triggered by the first;
    \item Sequentialization edges between nodes belonging to the same reactor, ordering them by logical time, with the reaction number as tiebreaker.
\end{enumerate}
Figure~\ref{fig:dag_gen} shows the DAG generated from the periodic phase of
Figure~\ref{fig:state_space_diagrams}. 

As an example for all the added edges, we focus on the node $c_2$, which is the second reaction of the controller reactor. It received the attributes $O_{c_2}=0$, $D_{c_2}=150$, and $C_{c_2}=15$ from the \lf program. The blue edges from $s_1$ and to $s_4$ set the bounds for the timing attributes. The red edge from $g_1$ is added because the first reaction of the gyroscope triggers this reaction. The green edge from $c_{1,1}$ is added because the first instance of the first reaction of the controller reactor has the same logical tag but lower number. The outgoing green edge to $c_{1,2}$ is set because it has a higher logical tag. Lastly, the red outgoing edge to $c_3$ is added because $c_3$ is triggered by $c_2$ through the logical action. An example for the need of the virtual path can be seen when inspecting the node $c_3$. Due to the logical action with a minimum delay of 100 the offset of the reaction should be $O_{c_3}=100$, necessitating the constraint $t^\text{EST}_{c_3} \geq 100$. By adding the edge from the sync node $s_3$, $c_3$ cannot start before 100 due to the total time required by $d_1$ and $d_2$.

\subsection{Scheduling by Graph Partitioning}
After the conversion to DAGs, the \lf compiler selects a user-specified quasi-static
scheduler and feed the generated DAGs to it.
At a high level, a quasi-static scheduler takes a set of 
graphs as input, partitions the graphs based on the number of workers
available, and outputs graph partitions, which are considered as
\term{quasi-static schedules}--a spatial mapping between tasks and execution
platforms with timing and inter-task dependencies.
The schedules are quasi-static because they statically encode tasks that
\emph{could} be triggered, but whether they are actually triggered  
is determined during execution using virtual instructions facilitating control flow.
More details on these instructions are explained in
Sec.~\ref{sec:static_schedules_to_pretvm}.

Given an unparititioned DAG, a quasi-static scheduler is free to designate DAG
nodes to any partition it deems fit. The scheduler can modify edges as long
as the output partitioned DAG satisfies all dependencies in the input
unpartitioned DAG.
In addition, the quasi-static scheduler must linearize each parition by adding
edges such that there exists a path in the output DAG passing through all nodes
within the same partition.

In this work, we support two different types of schedulers:
\textsc{Load Balanced} and \textsc{Edge Generation}~\cite{sun2023edge}.
The \textsc{Load Balanced} scheduler is a complete scheduler
implemented in Java, embedded inside the \lf compiler.
It directly takes the generated DAGs as input and
produces graph partitions that aim to assign each worker the same amount of
work measured in execution times.

The \textsc{Edge Generation} scheduler, on the other hand, is an external static
scheduler implemented in Python by Sun et al.~\cite{sun2023edge}. Instead of
distributing workloads fairly across workers, EGS focuses on
satisfying timing constraints. Given a DAG task with a deadline, EGS checks
whether the DAG task is \textit{trivially schedulable} on the given number of processors. 
A DAG task $(\mathcal{G}, P)$ is trivially schedulable on $M$ processors if its length is less than or equal to the period $P$ and its width is less than or equal to the number of processors $M$.
If it is not trivially
schedulable by construction, EGS adds edges trying to make it trivially
schedulable. A trivially schedulable DAG can easily be partitioned using paths
that cover the DAG, with an example shown in Figure~\ref{fig:dag_gen_red}.

\section{From Quasi-Static Schedules to \pretvm}
\label{sec:static_schedules_to_pretvm}
Once quasi-static schedules are generated in the form of graph partitions, they will
then be compiled into bytecode and executed by the \pretvm runtime. This section
walks through the green boxes in Figure~\ref{fig:pretvm_methodology} in
detail.
The main challenges we address here include composing virtual instructions to deliver
correct \lf semantics (\ie deterministic concurrency) and accounting for \lf's generality over LET (\ie
encoding the logical timeline and detecting whether a reaction should be
invoked during execution), while maintaining low overhead. 

\subsection{Virtual Instruction Set}

We begin by defining a virtual instruction set for \pretvm, which is used to
encode the quasi-static schedules. The full instruction set is shown in
Table~\ref{table:instruction_set}. 

\begin{table}[t]
  \centering
  \caption{\pretvm Virtual Instruction Set}
  \label{table:instruction_set}
  \scriptsize
  \bgroup
  \def\arraystretch{1.5}%
  \begin{tabularx}{\columnwidth}{|lX|}
  \hline
  \textbf{Instruction} & \textbf{Description} \\
  \hline
  \opcode{ADD  op1, op2, op3}   & Add to an integer variable (op2) by an integer variable (op3) and store to a destination register (op1).\\
  \hline
  \opcode{ADDI op1, op2, op3}  & Add to an integer variable (op2) by an immediate (op3) and store to a destination register (op1).\\
  \hline
  \opcode{ADV  op1, op2, op3} & Advance the logical time of a reactor (op1) to a
  base time register (op2) + a time increment variable (op3). \\
  \hline
  \opcode{ADVI op1, op2, op3} & Advance the logical time of a reactor (op1) to a base time register (op2) + an immediate value (op3). \\
  \hline
  \opcode{BEQ  op1, op2, op3} & Take the branch (op3) if the op1 register value is equal to the op2 register value. \\
  \hline
  \opcode{BGE  op1, op2, op3} & Take the branch (op3) if the op1 register value
  is greater than or equal to the op2 register value. \\
  \hline
  \opcode{BLT  op1, op2, op3} & Take the branch (op3) if the op1 register value
  is less than the op2 register value. \\
  \hline
  \opcode{BNE  op1, op2, op3} & Take the branch (op3) if the op1 register value
  is not equal to the op2 register value. \\
  \hline
  \opcode{DU   op1, op2} & Delay until the physical clock reaches a base timepoint (op1) plus an offset (op2). \\
  \hline
  \opcode{EXE  op1, op2} & Execute a function (op1) with argument (op2). \\
  \hline
  \opcode{JAL  op1 op2} & Store the return address to op1 and jump to a label (op2). \\
  \hline
  \opcode{JALR op1, op2, op3} & Store the return address to op1 and jump to a
  base address (op2) + an immediate offset (op3). \\
  \hline
  \opcode{STP} & Stop the execution. \\
  \hline
  \opcode{WLT  op1, op2} & Wait until a register value (op1) to be
  less than a desired value (op2). \\
  \hline
  \opcode{WU   op1, op2} & Wait until a register value (op1) to be greater than
  or equal to a desired value (op2). \\
  \bottomrule
  \end{tabularx}
  \egroup
\end{table}

The \pretvm is a register-based virtual machine.
The virtual instruction set borrows inspiration from the RISC-V instruction
set~\cite{waterman2014risc} and the timing instructions from the PRET
Machines~\cite{lickly2008cases,zimmer2014flexpret,jellum2023interpret}. 
The virtual instruction set contains standard RISC-V-like instructions such as
\opcode{ADD}, \opcode{ADDI}, \opcode{BEQ}, \opcode{BGE}, \opcode{BLT},
\opcode{BNE}, \opcode{JAL}, and \opcode{JALR}. These instructions are useful for
encoding control flow and manipulating auxiliary registers (further explained in
Sec.~\ref{subsec:generate_bytecode}) that keeps track of the progress of a
quasi-static schedule.

Besides the standard instructions, the virtual instruction set also contains
timing instructions, including \opcode{DU}, \opcode{WU}, and \opcode{WLT}. The
use of timing instructions adopts the PRET philosophy of making timing a
\emph{semantic} property of the instruction set. The \opcode{DU} instruction
helps align an execution with real time, and the \opcode{WU} and \opcode{WLT}
instructions are useful for implementing synchronization among workers.

Lastly, we introduce reactor-inspired instructions: \opcode{ADV},
\opcode{ADVI}, \opcode{EXE}, and \opcode{STP}. The \opcode{ADV}
and \opcode{ADVI} instructions advance the timestamp of a reactor's state.
\opcode{EXE} executes a C function pointer, which could be a user-written
reaction body or any functions required by the runtime to deliver the \lf semantics.
\opcode{STP} terminates the \lf program execution.

\subsection{Generating Bytecode from Quasi-Static Schedules}
\label{subsec:generate_bytecode}

With a virtual instruction set defined, we are almost ready to encode the quasi-static
schedules, \ie partitioned graphs, into bytecode, which consists of sequences of virtual
instructions, one for each worker thread.
Before we dive into the main instruction generation algorithm, let us walk
through a few prerequisites for encoding the quasi-static schedules using the virtual
instruction set.

\paragraph{Auxiliary Registers}
Encoding the semantics of quasi-static schedules requires storing certain information
at runtime in registers. Certain types
of registers are worker-specific, meaning that they duplicated for each
worker, while others are global, meaning that they are unique and read
by all workers.  
The progress of a worker during the execution of
its schedule is tracked using a worker-specific \register{counter} register.
A global \register{start\_time} register keeps track of the timestamp at which the
execution starts, while a global \register{timeout} register stores the last timestamp
(\ie the timestamp in the shutdown phase in
Figure~\ref{fig:state_space_diagrams}) of the execution, as specified by the user.
We use a global \register{time\_offset} register to keep track of the timestamp at
the beginning of the current hyperperiod.
To increment \register{time\_offset} by a variable amount, as the
execution goes through different phases, a global \register{offset\_inc} register
is used to store these variable increments at different points during the
execution. To store the return address of a worker before a jump, we use a
worker-specific \register{return\_addr} register. Moreover, for the purpose of
global synchronization at the boundary of a hyperperiod (explained later), each
a worker uses a \code{binary\_sema} register as a binary semaphore.
Lastly, there is a global constant \code{zero} register, which always holds a
value of 0.

\paragraph{Synchronization}
In our execution strategy, we support two types of synchronization: (i)
\term{local synchronization} between a pair of workers, and (ii) \term{global synchronization} across all workers.
Local synchronization is achieved using the \opcode{WU} and the \register{counter}
registers. Without loss of generality, 
assume that worker $B$ needs to wait 
until worker $A$ completes some upstream task. To implement this, in worker
$A$'s bytecode, we use an \opcode{ADDI} to increment worker $A$'s
\register{counter} after it completes the upstream task, and in worker
$B$'s bytecode, we use a \opcode{WU} to let it block until worker $A$'s
\register{counter} reaches the incremented value.
On the other hand, global synchronization across all workers is achieved by
inserting a \term{synchronization code block} into each worker's bytecode.
One worker is designated as the \term{coordinator}, denoted as $c$, facilitating the global
synchronization, while the other workers are \term{participants}. Each
participant is denoted as $p$. The coordinator executes the following code block:
\begin{enumerate}
  \item \label{emu:wu1} \opcode{WU} $\texttt{binary\_sema}_p, 1$
  \item \label{emu:wu2} (... Repeat \opcode{WU} for each participant $p$ ...)
  \item \label{emu:update_offset} \opcode{ADD} $\texttt{time\_offset}, \texttt{time\_offset}, \texttt{offset\_inc}$
  \item \label{emu:reset_counter1} \opcode{ADDI} $\texttt{counter}_c, \texttt{zero}, 0$
  \item \label{emu:reset_counter2} \opcode{ADDI} $\texttt{counter}_p, \texttt{zero}, 0$
  \item \label{emu:reset_counter3} (... Repeat \opcode{ADDI} for each participant $p$ ...)
  \item \label{emu:adv1} \opcode{ADVI} $\textit{reactor}, \texttt{time\_offset}, 0$
  \item \label{emu:adv2} (... Repeat \opcode{ADVI} for each \textit{reactor} in the \lf program ...)
  \item \label{emu:addi1} \opcode{ADDI} $\texttt{binary\_sema}_p, \texttt{zero}, 0$
  \item \label{emu:addi2} (... Repeat \opcode{ADDI} for each participant $p$ ...)
  \item \label{emu:jalr} \opcode{JALR} $\texttt{zero}, \texttt{return\_addr}_c, 0$
\end{enumerate}
Conversely, each participant executes this shorter code block:
\begin{enumerate}
  \item \label{emu:addiP} \opcode{ADDI} $\texttt{binary\_sema}_p, \texttt{zero}, 1$
  \item \label{emu:wltP} \opcode{WLT} $\texttt{binary\_sema}_p, 1$
  \item \label{emu:jalrP} \opcode{JALR} $\texttt{zero}, \texttt{return\_addr}_p, 0$
\end{enumerate}
During global synchronization, the coordinator first waits until all
participants have entered their own sync. blocks
(line~\ref{emu:wu1}-\ref{emu:wu2}). A participant enters the sync block and
notifies the coordinator by setting its \register{binary\_sema} to 1
(line~\ref{emu:addiP} in the participant code below).
Once they have entered and get blocked by the \opcode{WLT} 
(line~\ref{emu:wltP} below), the coordinator then
updates the global hyperperiod offset (line~\ref{emu:update_offset} above),
resets all workers' \register{counter} registers
(line~\ref{emu:reset_counter1}-\ref{emu:reset_counter3}), and advances the
timestamp of all reactors' states (line~\ref{emu:adv1}-\ref{emu:adv2}).
Once these ``bookkeeping'' procedures are done, the coordinator resets the
\register{binary\_sema} registers to 0 (line~\ref{emu:addi1}-\ref{emu:addi2}
above), at which point the participants are unblocked by \opcode{WLT} and jump
back to the return addresses specified in their \register{return\_addr}
registers (line~\ref{emu:jalrP} below). 

\paragraph{Connections}
\lf allows the user to specify connections with optional fixed logical delays.
When the delays are greater than the rate at which output ports are written, a
storage layer is required to store the in-flight events. In this work, we
implement a connection buffer for each delayed connection using a fixed-size
circular buffer. Since the \lf semantics guarantees that events sent along a
connection have monotonically increasing timestamps, there is no need to reorder
events and hence circular buffers suffice.

For each connection, we code-generate two \term{connection helper functions}:
a pre-connection helper function pushes an event into the circular buffer, and
a post-connection helper function pops an event from the circular buffer. The
quasi-static schedule executes connection helper functions using the
\opcode{EXE} instruction.
% Where we generate them.
Since the \lf semantics allows an output port to be written to by multiple
reactions triggered at the same timestamp, with the last reaction taking
precedence, a pre-connection helper is inserted after invoking the last
reaction that could modify this port to push the latest written value into the
connection buffer.
If a reaction is triggered by a port, a post-connection helper is inserted to
pop the stale value read by the reaction from the connection buffer.

% Optimization
If a connection has no delay, we optimize the connection helper functions away
and let the downstream receiving reaction read value directly from the pointer to the
sending reactor's output port.

\paragraph{Trigger Detection}
The \lf semantics requires a reaction to be invoked when \textit{any} of its
triggers becomes present. Since it is unknown at compile time whether an
input-triggered reaction does get triggered during execution, the quasi-static
schedule must encode logic for checking the presence of triggers at runtime.
We implement the trigger detection mechanism using a combination of
\opcode{BEQ}, \opcode{JAL}, and \opcode{EXE}. 
For every trigger a reaction has that is an input port, a \opcode{BEQ} is
generated. If the input port is connected to a delayed connection, the
\opcode{BEQ} checks if the head event of the connection buffer has the same
timestamp as the current timestamp of the reaction's parent reactor. Otherwise
if the input port is connected to a zero-delay connection, the \opcode{BEQ}
checks if the sending output port's \register{is\_present} field is \code{true}.
In both cases, when the condition evaluates to true, the \opcode{BEQ} branches
to the location of an \opcode{EXE} instruction invoking the reaction body.
After a list of \opcode{BEQ} is generated, a \opcode{JAL} instruction is
generated, which branches to the instruction \textit{after} the
reaction-invoking \opcode{EXE}, when all triggers are absent, in which case the
reaction-invoking \opcode{EXE} is bypassed.
A concrete example of this mechanism is shown in
Figure~\ref{fig:bytecode_example} (line~\ref{line:test_trigger}-\ref{line:execute_reaction}).

\paragraph{Instruction Generation}
Having discussed the above prerequisites for generating bytecode that encodes
quasi-static schedules, we are now ready to dive into 
Algorithm~\ref{algo:generate_instructions}, which shows the procedure for
generating virtual instructions from a partitioned graph.

\begin{algorithm}
  \caption{Generate Instructions from a Partitioned Graph}
  \small
  \label{algo:generate_instructions}
  \begin{algorithmic}[1]
  \Procedure{generateInstructions}{$\textit{graphParitioned}$}
      \State Let $B_i$ be the set of instructions for worker $i$
      \State Let $W$ be the total number of workers
      \State Topologically sort the DAG nodes by dependencies
      \ForAll{$\text{currentNode} \in \text{topologicalSort}$}
          \State $w \gets \text{currentNode.assignedWorker}$
          \If{$\text{currentNode.type} = \textit{reaction}$}
              \ForAll{upstream reactions of currentNode}
      \State $B_w \gets B_w \cup \{\texttt{WU}\}$ \Comment{Wait for upstream}
              \EndFor
              \If{currentNode depends on \textit{sync} node}
                \State \textsc{GeneratePreConnectionHelpers()}  
                \State $B_w \gets B_w \cup \{\texttt{ADVI}, \texttt{DU}\}$
              \EndIf
              \ForAll{input port triggers of currentNode}
      \State $B_w \gets B_w \cup \{\texttt{BEQ}\}$ \Comment{Test trigger presence}
              \EndFor
              \State $B_w \gets B_w \cup \{\texttt{JAL}\}$ \Comment{Absent triggers, skip
      \texttt{EXE}}
              \State $B_w \gets B_w \cup \{\texttt{EXE}\}$ \Comment{Execute reaction}
              \State \textsc{GeneratePostConnectionHelpers()} 
              \State $B_w \gets B_w \cup \{\texttt{ADDI}\}$ \Comment{Update counter}
              \ElsIf{$\text{currentNode.type} = \text{\textit{sync}}$ and
              currentNode is tail node of DAG}
              \State \textsc{GeneratePreConnectionHelpers()}
              \If{currentNode.time is not $\infty$}
                  \For{$i = 0$ to number of workers - 1}
                      \State $B_i \gets B_i \cup \{\texttt{DU}\}$ \Comment{End of hyperperiod}
                      \If{$i = 0$}
      \State $B_i \gets B_i \cup \{\texttt{ADDI}\}$ \Comment{Iterate again}
                      \EndIf
                      \State $B_i \gets B_i \cup \{\texttt{JAL}\}$ \Comment{Synchronize}
                  \EndFor
              \EndIf
          \EndIf
      \EndFor
      \State \Return $\bigcup_{i=0}^{W - 1} B_i$
  \EndProcedure
  \end{algorithmic}
\end{algorithm}

The algorithm first topologically sorts the nodes of the
partitioned graph and traverses the graph in the sorted order, ensuring that
before visiting a node, all of its dependencies have been already visited.
If the current node is a reaction node, then a \opcode{WU} is
generated for each of its upstream tasks. Moreover, if the reaction node depends on
a \textit{sync} node, pre-connection helper functions are first generated for
earlier reactions that write to output ports; then, \opcode{ADVI} and \opcode{DU}
instructions are generated to advance the timestamps of the parent reactor and
delay the execution until the physical release time of the reaction node. 
Then, instructions for trigger detection are generated, which start with a
sequence of \opcode{BEQ} instructions and a \opcode{JAL} instruction.
This is then followed by an \opcode{EXE} instruction
for executing the reaction body. 
Once the reaction is executed, a sequence of \opcode{EXE} instructions execute
the post-connection helper functions.
An \opcode{ADDI}
instruction is then used to increment a worker counter, effectively updating the
current progress of the worker.
If the current node is an \textit{sync} node and is the tail node of the partitioned
graph, then pre-connection helper functions are generated for all the remaining
reactions invoked in the current hyperperiod, which can write to output ports.
A \opcode{DU} is also added to each worker's bytecode. Each worker
further uses an \opcode{ADDI} to increment the hyperperiod register. Once that
is done, each worker executes a \opcode{JAL} to the synchronization block.

\paragraph{Linking}
Once a sequence of instructions is generated for each partitioned graph for each
worker, the Bytecode Generator links multiple sequences together into a single
piece of output bytecode. The linking process proceeds by phases. It starts with the
first phase the execution enters into, such as the initialization
phase, and places the generated instructions from this phase into the output
bytecode. Then, the process follows the guarded transitions and places the
generated instructions from downstream phases one-by-one, until there are no
more unlinked phases with incoming transitions. At the end of the instructions
generated by each phase, the Bytecode Generator places instructions generated
from guarded transitions using \opcode{JAL} and branch instructions. In the
running example (Figure~\ref{fig:state_space_diagrams}), the default transition
between the initialization phase and the periodic phase is encoded as a
\opcode{JAL} instructions between two code blocks corresponding to the two
phases, similarly for the default self transition in the periodic phase. The
guarded transition with the guard $t \succeq \text{LAST TAG}$ is encoded as 
$\texttt{BGE} \; t, \text{LAST TAG}, \text{SHUTDOWN}$.

Finally, as a concrete example, the generated bytecode for the worker assigned
to the blue partition in Figure~\ref{fig:dag_gen_red} is shown in
Figure~\ref{fig:bytecode_example}.
\begin{figure}[h]
  \begin{enumerate}
    \item \texttt{EXE} $g_1$\label{line:g1_node_starts} \qquad\qquad(This line is labeled $\text{PERIODIC}_\text{blue}$.)
    \item \texttt{ADDI} $\texttt{counter}_\text{blue}$, $\texttt{counter}_\text{blue}$, $1$\label{line:g1_node_ends}
    \item \texttt{WU} $\texttt{counter}_\text{green}$, $2$
    \item \texttt{BEQ} $\texttt{in1\_is\_present}$, $\textit{true}$, \text{line\_\ref{line:execute_reaction}}\label{line:test_trigger}
    \item \texttt{JAL} \text{line\_\ref{line:pass_exe}}
    \item \texttt{EXE} $c_2$\label{line:execute_reaction}
    \item \texttt{EXE} $\texttt{out0\_pre\_connection\_helper}$\label{line:pre_conn_helper}
    \item \texttt{EXE} $\texttt{in1\_post\_connection\_helper}$\label{line:post_conn_helper}
    \item \texttt{ADDI} $\texttt{counter}_\text{blue}, \texttt{counter}_\text{blue}, 1$\label{line:pass_exe}
    \item \texttt{DU} $\texttt{time\_offset}, 150 \mu s$
    \item \texttt{ADDI} $\texttt{offset\_inc}, 150 \mu s$
    \item \texttt{JAL} $\texttt{return\_addr}_\text{blue}, \text{SYNC\_BLOCK}$
    \item \texttt{BGE} $\texttt{time\_offset}, \texttt{timeout}, \text{SHUTDOWN}_\text{blue}$
    \item \texttt{JAL} $\texttt{return\_addr}_\text{blue}, \text{PERIODIC}_\text{blue}$
  \end{enumerate}
  \caption{Generated bytecode for the blue worker in Figure~\ref{fig:dag_gen_red}}
  \label{fig:bytecode_example}
\end{figure}

\subsection{Executing Bytecode on \pretvm}
\label{subsec:execute_bytecode}
At this point in the workflow (Figure~\ref{fig:pretvm_methodology}), we have
finished generating the \pretvm bytecode.
To execute the bytecode, we implement a \pretvm Runtime, which implements the
virtual instruction set and interprets the bytecode during execution.
To make the \pretvm Runtime compatible with the \lf instrumentation code, we
embed the runtime under the scheduler API of the \lf runtime library. 
The \pretvm runtime, the instrumentation code, and the \pretvm bytecode is
further compiled by a C compiler, such as GCC, into an executable.

\section{Predictable Timing}
\label{sec:predictable_timing}
Predictability is a central concern in real-time system design. In the
literature, however, predictability is often subject to different
interpretations. We argue that whether a system is predictable depends on the
\textit{prediction} one aims to make.
Prior works have focused on finding the worst-case execution times (WCETs) of
individual tasks on various platforms.
However, finding a WCET at the system level, where tasks coordinate and execute concurrently on
multicore platforms, remains an open problem.

An \lf runtime based on \pretvm cleanly separates application logic
(\ie reaction bodies) from coordination logic (\ie quasi-static schedules), presenting a practical way of writing
real-time software that is amenable to timing analysis. 
On the one hand, \lf reactions are run-to-completion at the time of writing, removing the need to account for
preemption and thus simplifying timing analysis.
On the other hand, \pretvm instructions' straightforward semantics simplifies the WCET
analysis of each instruction.

Here we make a key assumption that the user can derive WCET values for each reaction body and
every \pretvm instruction on a given platform, \eg using timing analysis tools
or direct measurements. 
We treat these WCET values as \textit{estimates}---they do not need to be ground
truth, but violations of them are considered fault conditions handled at runtime. 
Assuming that the individual WCET estimates are not violated during execution,
we can now predict
the WCET of an \lf program's hyperperiod, represented as a DAG and implemented in \pretvm bytecode.

We denote a time domain as $\mathbb{T}$. In practice, this domain is defined as
$\mathbb{T} = \mathbb{Z}$, where the unit is nanosecond. Based on the previous
assumption, the user can obtain WCETs for individual reaction bodies and \pretvm
instructions, i.e., there exists a function $w_r : \mathcal{R} \rightarrow
\mathbb{T}$ and a function $w_i : \mathcal{I} \rightarrow
\mathbb{T}$, where $\mathcal{R}$ denotes the set of reactions and the
$\mathcal{I}$ denotes the set of \pretvm instructions.
For the WCETs of \opcode{DU}, \opcode{WU}, and \opcode{WLT}, their wait time
does not count toward the WCETs, only the overhead of calling these instructions
does. Similarly for \opcode{EXE}, the WCET of the function its first operand points to
does not count toward the WCET of the \opcode{EXE} instruction itself.

\begin{definition}[WCET of a node]
  The WCET of a node $n$ is a function $w : \mathcal{V} \rightarrow \mathbb{T}$ defined as
  \[
  w(n) = 
  \begin{cases} 
  0 & \text{if } $n$ \text{ is \textit{sync}}, \\
  d & \text{if $n$ is \textit{dummy} with interval $d$}, \\
  \mathlarger{\sum}_{i \in I(n)} w_i(i) + w_r(r) & \text{if $n$ is for reaction $r$}.
  \end{cases}
  \]
  where $I : \mathcal{V} \rightarrow \mathcal{P}(\mathcal{I})$ returns the set
  of \pretvm instructions generated by a given node $n$.
\end{definition}
Referring back to the running example, node $g_1$ in Figure~\ref{fig:dag_gen_red}
generates line~\ref{line:g1_node_starts}-\ref{line:g1_node_ends} in
Figure~\ref{fig:bytecode_example}, so when the user provides a WCET annotation
for the reaction in \reactor{Gyroscope}, the annotated WCET estimate should account for
\opcode{EXE}, \opcode{ADDI}, and the reaction body in \reactor{Gyroscope}.

\begin{lemma}[WCET up to node $n$]\label{lem:wcet_up_to_n}
  The WCET up to a node $n$, denoted recursively as $\bar{w}(n)$, is defined as
  \[ \bar{w}(n) = 
  \begin{cases} 
    w(n) & \text{if } U(n) = \varnothing, \\
    \max_{u \in U(n)}\bar{w}(u) + w(n) & \text{otherwise}.
  \end{cases}
  \]
  where $U : \mathcal{V} \rightarrow \mathcal{P}(\mathcal{V})$ maps a DAG node $n$
  to a set of its immediately upstream DAG nodes, and if $n$ is the tail node, $U(n)$
  excludes virtual nodes from the returned set.
\end{lemma}
\begin{proof}
  We prove this theorem by induction. In the base case, the DAG has a single \textit{sync}
  node $n$, which is also the tail node, and so its WCET, which is by
  construction 0, equals $ w(n) = \bar{w}(n) $ since the tail node has no
  upstream. The base case therefore holds. In the inductive step, the DAG has
  more than one nodes and the tail node necessarily has upstream nodes, \ie
  $U(n) \neq \varnothing$.
  By the induction hypothesis, for each upstream node $u \in U(n)$, $\bar{w}(u)$
  is the WCET up to node $u$. Since all upstream nodes point to the tail, a
  maximum of $\bar{w}(u)$ needs to be taken over all $u \in U(n)$. Given that
  the tail node $n$ has a WCET of 0, the WCET up to the tail node
  equals $\max_{u \in U(n)}\bar{w}(u) + 0 = \max_{u \in U(n)}\bar{w}(u) + w(n) =
  \bar{w}(n).$
  Therefore, the induction step holds.
\end{proof}

\begin{theorem}[WCET of a DAG task set]
  The WCET of a DAG with the tail node $n_t$ is $\bar{w}(n_t)$.
\end{theorem}
\begin{proof}
  The proof follows from Lemma~\ref{lem:wcet_up_to_n}.
\end{proof}

\begin{example}
  The WCET of the DAG in Figure~\ref{fig:dag_gen_red} can be captured by
  $\bar{w}(s_4)$, where
  \begin{align*}
    \bar{w}(s_4) & = \max(\bar{w}(a_{1,2}), \bar{w}(m_2)) + w(s4) \\
    & = \max(\bar{w}(s_2) + 20, \bar{w}(c_3) + 15) + 0 \\
    & = \max(\max(75, \bar{w}(c_{1,1})) + 20, \\ & \ \ \ \max(\bar{w}(s_3), \bar{w}(c_{1,2})) + 40) \\
    & = ... = 150
  \end{align*}
\end{example}

\section{Evaluation}

% Generated table at /Users/shaokai/Documents/projects/lingua-franca/lf-pubs/PretVM/artifacts/benchmarks/experiment-data/timing/2024-03-30_00-51-22-RPI4-LOAD_BALANCED_EGS-ALL/table.tex

    \begin{table*}[ht]
        \centering
        \begin{tabular}{lcccccccccc}
        \toprule 
        & & \multicolumn{3}{c}{Average (us)} & \multicolumn{3}{c}{Maximum (us)} &
        \multicolumn{3}{c}{Standard Deviation (us)} \\ 
        \cmidrule(lr){3-5} \cmidrule(lr){6-8} \cmidrule(lr){9-11}
        Program & LoC (\lf) & DY & LB & EG & DY & LB & EG & DY & LB & EG \\ 
        \midrule 
    
\texttt{Philosophers}  & 314  & 18.4 & 12.1 & 11.9 & 355 & 99.2 & 510 & 9.99 & 5.65 & 9.33 \\
\texttt{LongShort}  & 31  & 1.49e+06 & 3.11 & 3.17 & 2.91e+06 & 5.22 & 4.32 & 8.11e+05 & 0.781 & 0.571 \\
\texttt{CoopSchedule}  & 54  & 18.5 & 10.9 & 12.3 & 50.4 & 32.2 & 69 & 7.64 & 4.69 & 6.4 \\
\texttt{Counting}  & 179  & 14.6 & 6.81 & 6.74 & 76.8 & 42.4 & 51.5 & 5.95 & 2.74 & 2.99 \\
\texttt{ThreadRing}  & 217  & 19.5 & 14.8 & 12.3 & 44.2 & 58.3 & 36.2 & 7.18 & 8.05 & 6.28 \\
\texttt{ADASModel}  & 91  & 12.9 & 5.49 & 1.81e+03 & 41.6 & 19.4 & 1.2e+04 & 5.5 & 3.35 & 3.58e+03 \\
\texttt{PingPong}  & 124  & 13.6 & 8.02 & 7.96 & 64.4 & 59.9 & 56.8 & 4.14 & 3.95 & 3.87 \\
\texttt{Throughput}  & 166  & 16.2 & 13.6 & 14.1 & 29.1 & 33.4 & 30.5 & 5.18 & 6.87 & 7.04 \\
        \bottomrule
        \end{tabular} 
        \caption{Average, maximum, and standard deviation of the lags in microseconds of the
        dynamic scheduler (DY), the static \textsc{Load
        Balanced} scheduler (LB), and the static \textsc{Edge Generation}
        scheduler (EG).} 
        \label{tab:accuracy_results}
    \end{table*}

We evaluate the \textit{timing accuracy} of \pretvm by running a set of real-time \lf
benchmarks on Raspberry Pi 4 Model B running Linux.
In this experiment, accuracy is measured in terms of the \textit{lag} of
reaction invocations. For a reaction invocation $i$, its lag, $l_i$, is defined as
\[
l_i = T_i - t_i,
\]
where $T_i$ is the physical time at which the reaction invocation occurs and
$t_i$ the logical timestamp. Intuitively, the smaller the difference between the
physical invocation time and the ideal logical invocation time, the more
time-accurate an execution is. We run the same set of benchmarks using two static 
schedulers, \textsc{Load Balanced} and \textsc{Edge Generation}, and we compare
their performance against \lf's default \textsc{Dynamic Scheduler}.

The eight benchmarks are selected from~\cite{menard2023performance,lin2023lfverifier} with a focus on concurrent
actor design patterns~\cite{imam2014savina},
which can be found in modern Cyber-Physical Systems. For each benchmark program, we annotate
each reaction's WCET estimate using an \texttt{@wcet} attribute. Syntax elements
which are currently not supported by our implementation include nested reactors,
logical and physical actions, and multiports. A logical action is replaced by adding
a connection that links two newly created output and input ports. The reactor hierarchy
is flattened manually and multiports were expanded by hand.

% Data collection
We collect performance data using \lf's tracing utility, which embeds
trace points throughout the codebase to inspect notable events in the
runtime system. Each event contains an event type, a logical timestamp, and a
physical timestamp. In this experiment, we only collect events for the starting points
of reactions. This choice ensures that the overhead of tracing stays
uniform for all three schedulers under test. To avoid distortions in our measurements, 
we have also increased the size of the trace buffer to minimize data being written to disk 
during execution.

% Fairness
\lf schedulers need to wait for the physical time to surpass the next
reaction invocation's logical timestamp before processing the invocation. To ensure a fair comparison between the schedulers, we implement the wait
mechanism in all three schedulers under test using busy wait, replacing the
sleep mechanism found in the default dynamic scheduler. By doing this, we
remove timing variation from operating system's sleep mechanism and
attribute the difference in performance to scheduler designs.

\begin{figure*}
  \centering
    \begin{subfigure}[t]{0.45\textwidth}
        \centering
        \includegraphics[width=\textwidth]{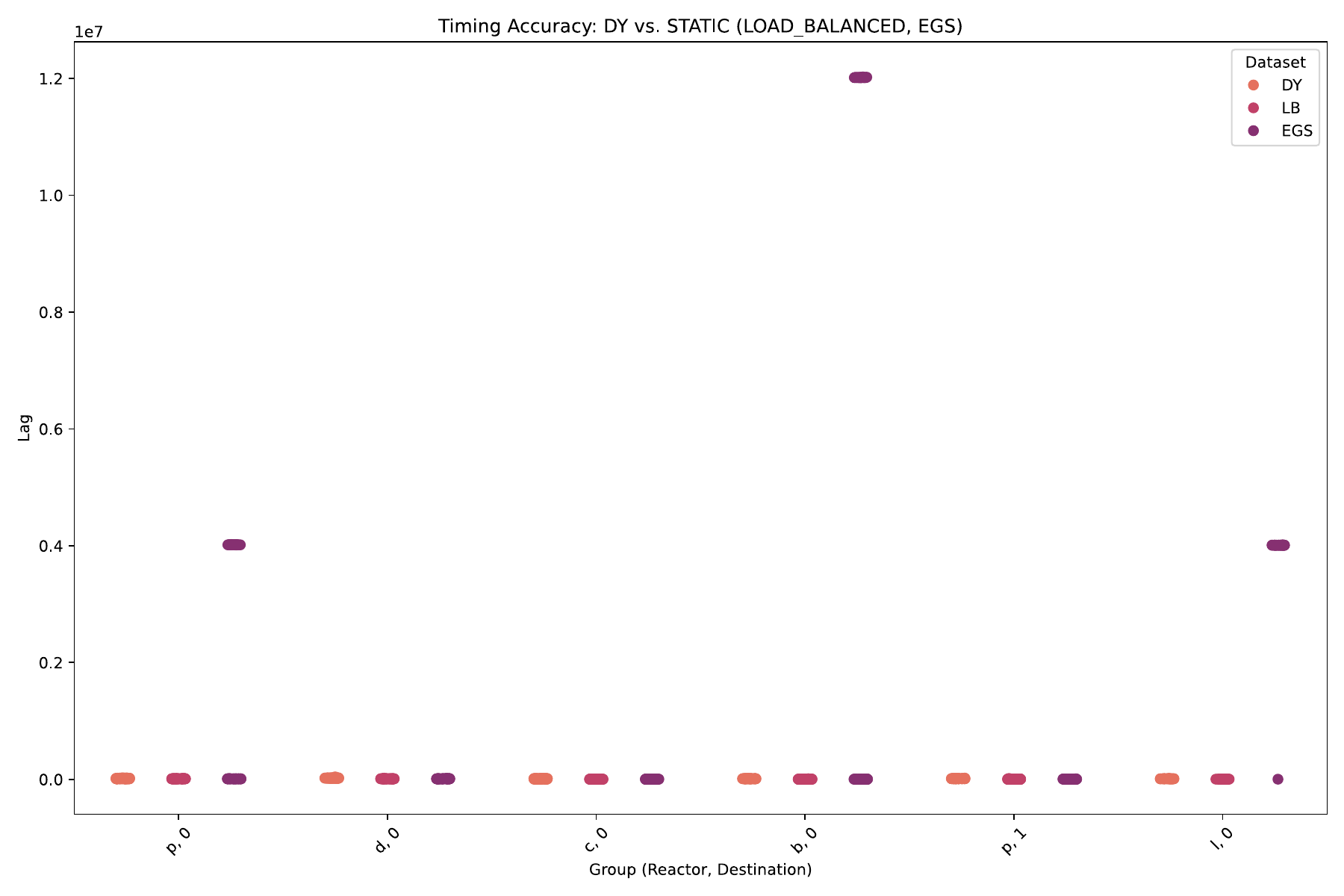}
        \caption{\texttt{ADASModel}}
    \end{subfigure}%
    ~ 
    \begin{subfigure}[t]{0.45\textwidth}
        \centering
        \includegraphics[width=\textwidth]{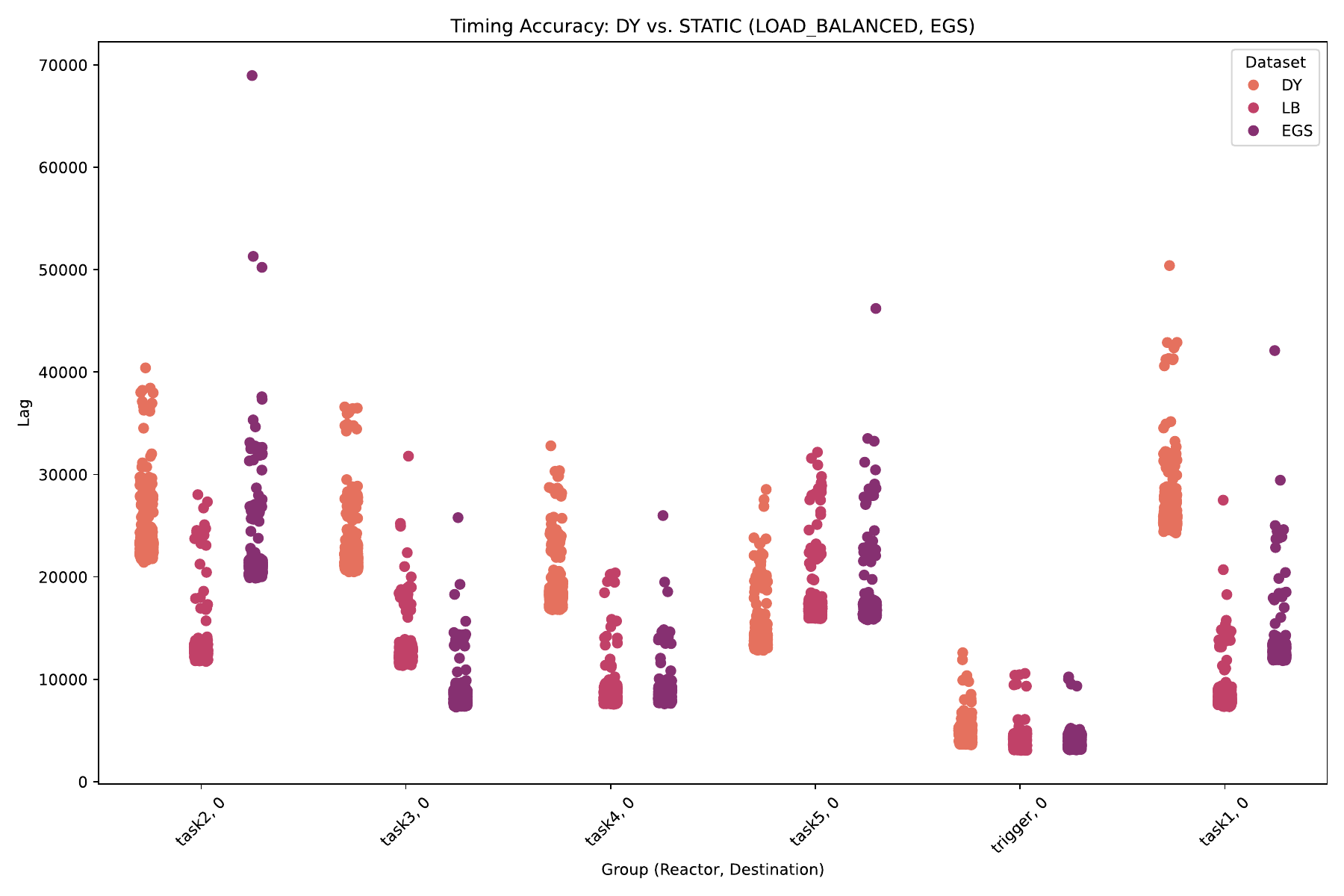}
        \caption{\texttt{CoopSchedule}}
    \end{subfigure}
    ~ 
    \begin{subfigure}[t]{0.45\textwidth}
        \centering
        \includegraphics[width=\textwidth]{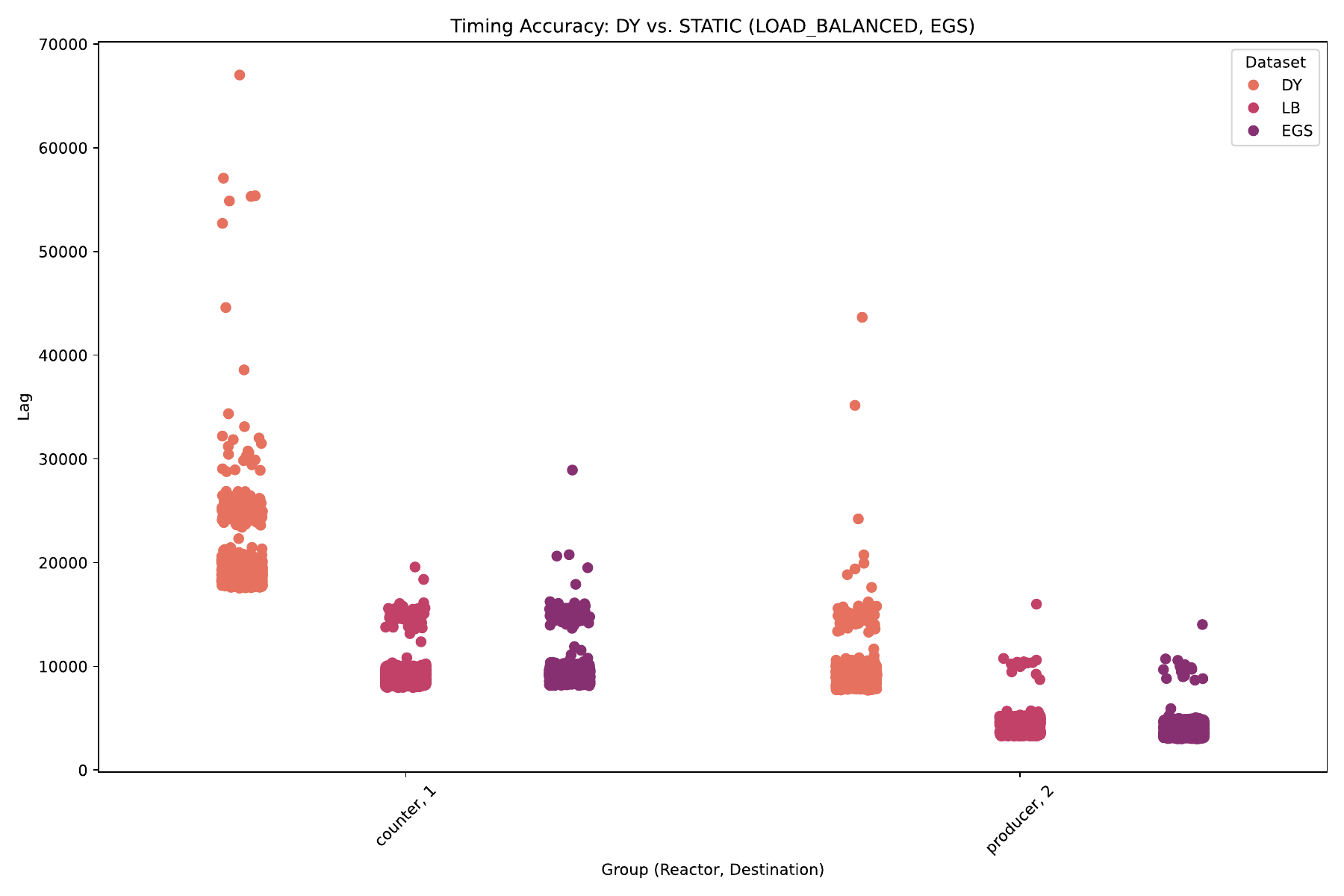}
        \caption{\texttt{Counting}}
    \end{subfigure}
    ~ 
    \begin{subfigure}[t]{0.45\textwidth}
        \centering
        \includegraphics[width=\textwidth]{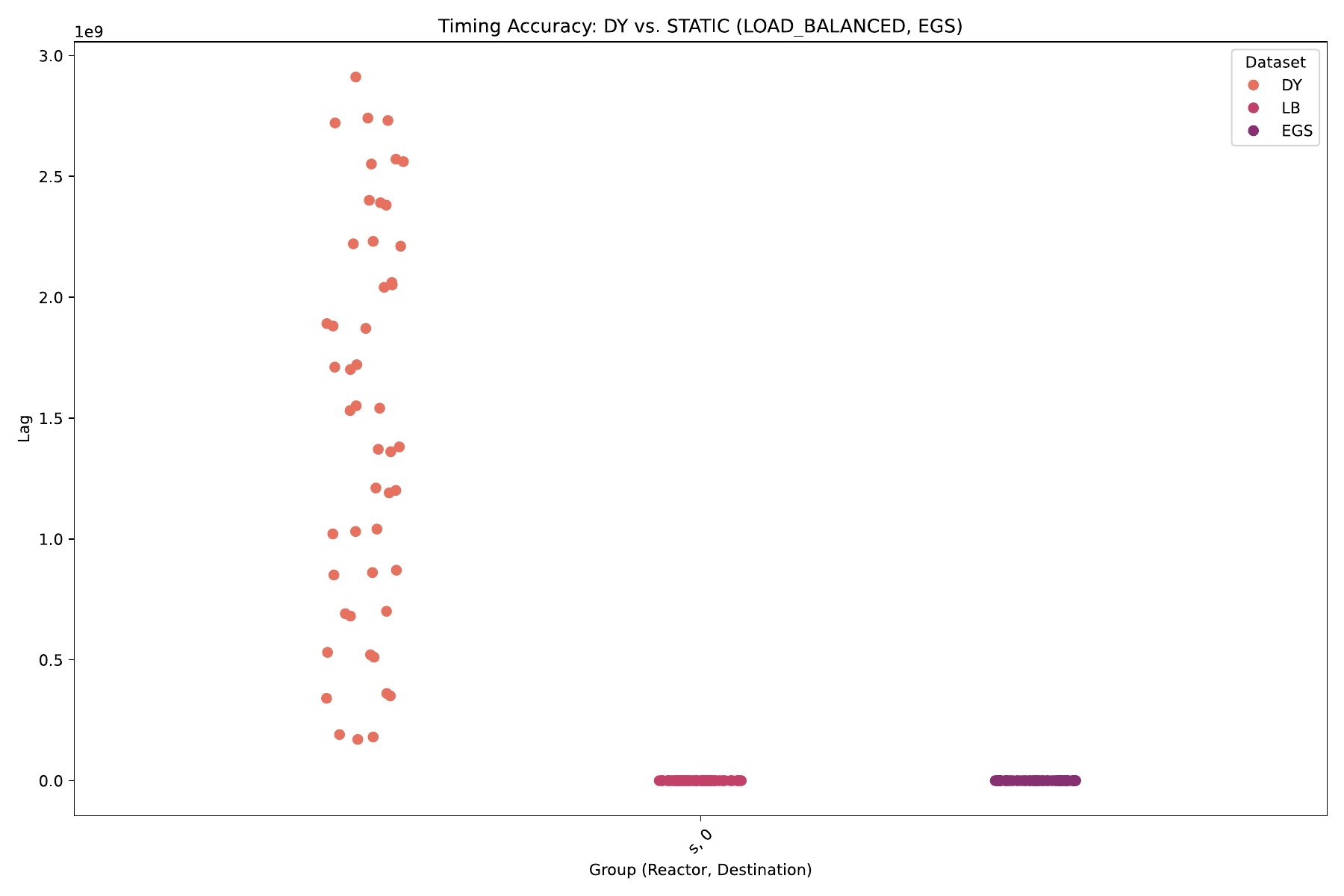}
        \caption{\texttt{LongShort}}
    \end{subfigure}
    ~ 
    \begin{subfigure}[t]{0.45\textwidth}
        \centering
        \includegraphics[width=\textwidth]{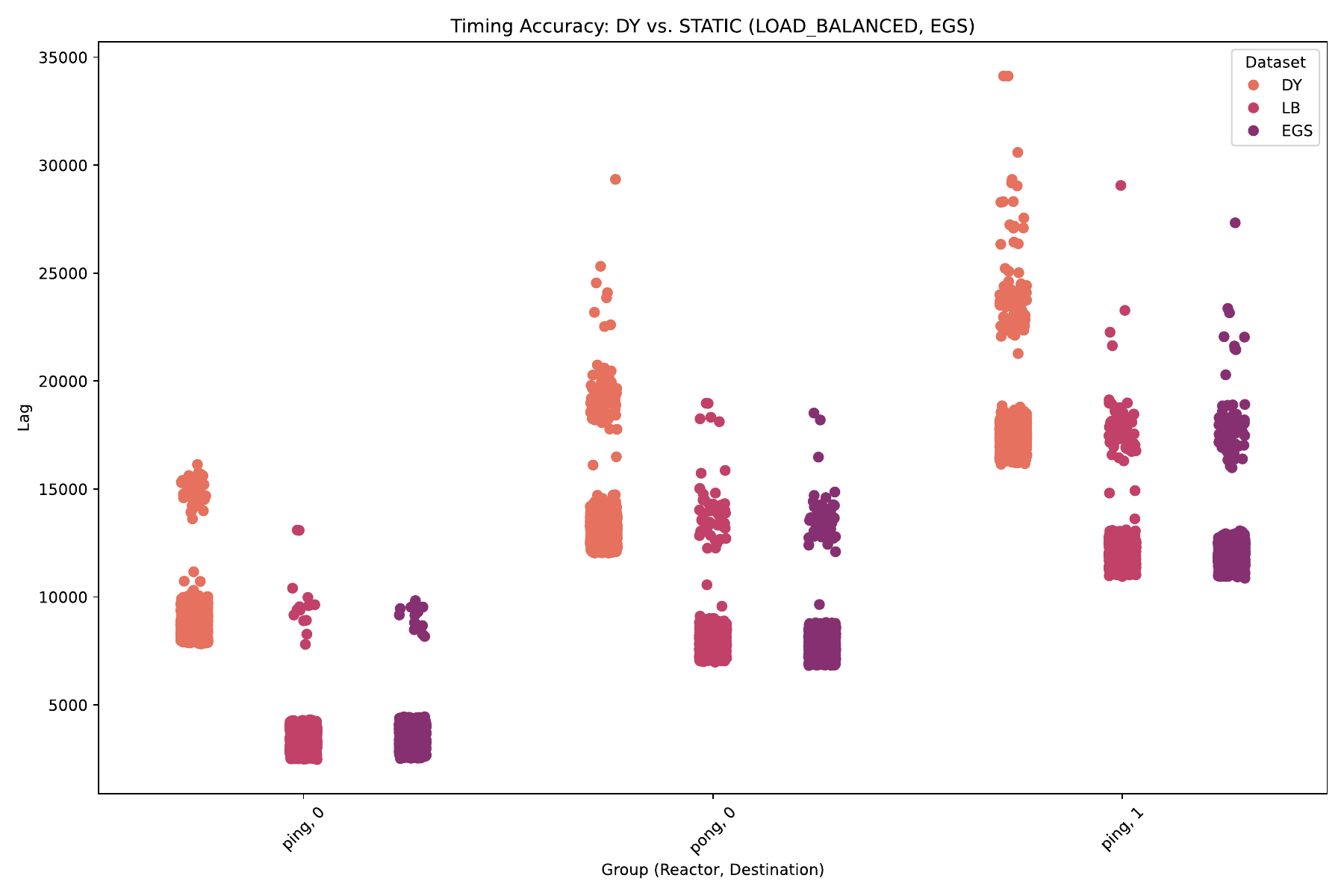}
        \caption{\texttt{PingPong}}
    \end{subfigure}
    ~ 
    \begin{subfigure}[t]{0.45\textwidth}
        \centering
        \includegraphics[width=\textwidth]{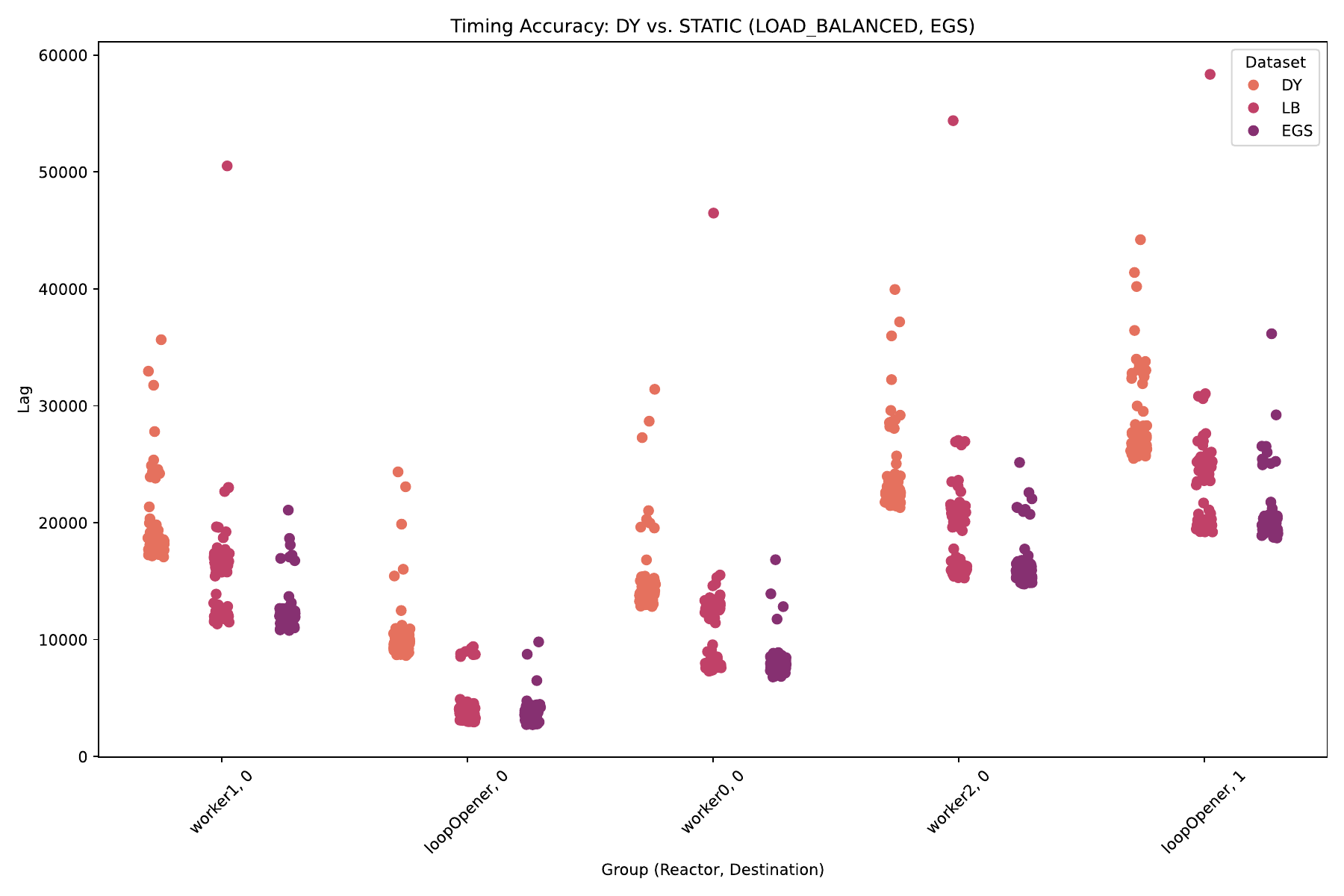}
        \caption{\texttt{ThreadRing}}
    \end{subfigure}
    ~ 
    \begin{subfigure}[t]{0.45\textwidth}
        \centering
        \includegraphics[width=\textwidth]{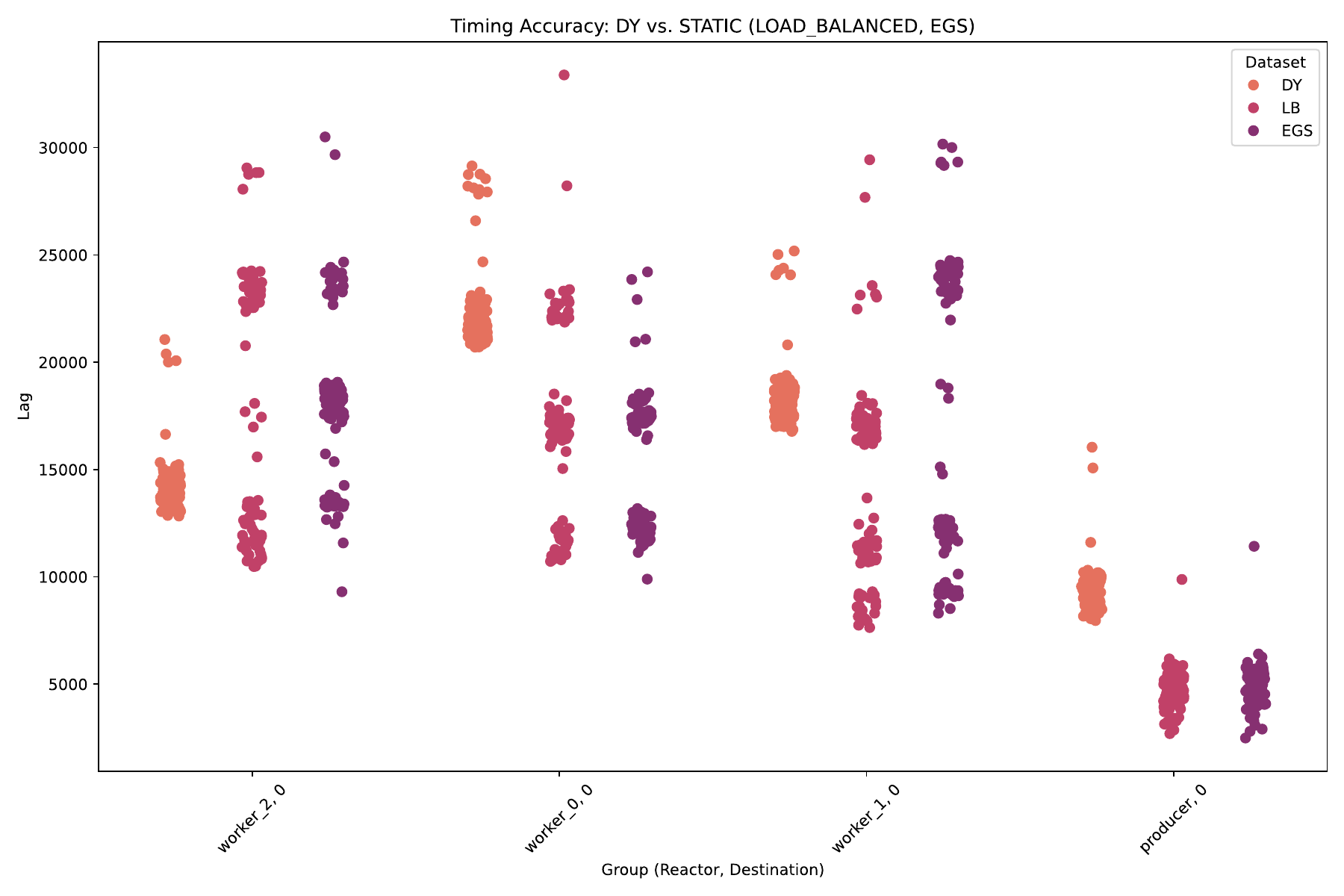}
        \caption{\texttt{Throughput}}
    \end{subfigure}
    ~ 
    \begin{subfigure}[t]{0.45\textwidth}
        \centering
        \includegraphics[width=\textwidth]{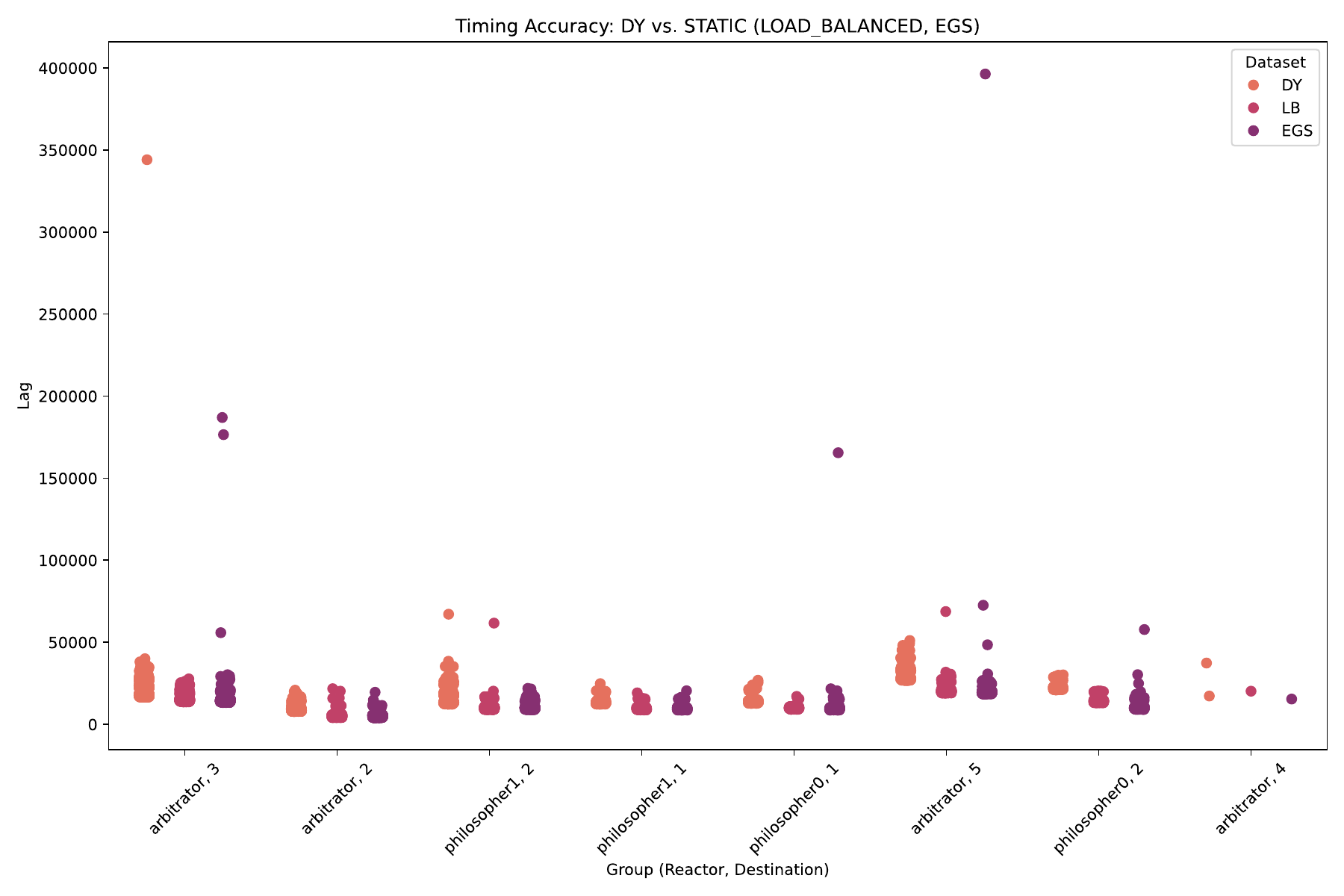}
        \caption{\texttt{Philosophers}}
    \end{subfigure}
    \caption{Timing accuracy benchmark results. X-axis shows reaction names. Each reaction contains three distributions, DY, LB, and EG, respectively.  Y-axis is lag in nanoseconds.}
    \label{fig:accuracy-plots}
\end{figure*}

Table~\ref{tab:accuracy_results} shows an overview of the timing accuracy
results. Compared to the dynamic scheduler, the static \lb
scheduler has smaller average and maximum lags on all benchmarks, and it has a
smaller standard deviation on all benchmarks except \texttt{ThreadRing}.

Figure~\ref{fig:accuracy-plots} shows strip plots of lags during the
periodic phase of
each program.
For most reactions, \lb and \egs have smaller lags than the dynamic scheduler.
The occasional outliers shown in the strip plots can be due to several reasons,
including interrupts from the underlying OS and the flushing of trace buffers to
disk.

In \texttt{ADASModel}, \egs shows three clusters of outliers. This is due to
\egs's scheduling strategy involving adding edges to the DAG, which tends to
delay the releasing of tasks in exchange for schedulability guarantees. We
note that the goal of using \egs is not to minimize lags, but to use its analysis
capability to ensure that a
DAG is schedulable. In this case, larger lags from \egs seem tolerable. In
general, \lb is better at miniming lags than \egs and the dynamic scheduler.

Notably, for the \texttt{LongShort} benchmark, \lb and \egs
significantly outperform the dynamic scheduler.
This performance gain can be explained by the fact that the quasi-static schedulers
keep track of the current logical time for \textit{each} reactor,
meaning that every reactor can advance its logical time independently as long as the
core \lf semantics is respected, \ie each reactor processes events in timestamp order.
The dynamic scheduler, on the other hand, uses a single variable to track
logical time for \textit{all} reactors, ensuring that all reactors advance time
together. The \texttt{LongShort} benchmark has a design pattern that combines
infrequent, long-running reactions with frequent, short-running reactions. The
dynamic scheduler struggles because it performs a \term{global barrier synchronization}
at the end of each tag. The barrier makes reactions that are to occur at the
next tag wait on the long-running reactions, even if there are no more actual
data dependencies to fulfill.
In the quasi-static schedulers, the capability to advance time
individually for each reactor effectively prevent waiting unnecessarily,
resulting in significantly reduced lags. 
We plan to compare our \pretvm approach with ongoing research on easing global
barrier synchronization in dynamic schedulers in the future.

\section{Related Work}
Our work is largely inspired by the Embedded
Machine~\cite{Henzinger:02:EMachine}, which serves as an execution backend for
the Giotto language~\cite{Henzinger:01:Giotto,Henzinger:03:Giotto}, implementing
the LET model~\cite{kirsch2012let}. 
An emerging extension of LET is the
System-level LET (SL LET)
model, which scales LET to distributed settings by introducing a notion
of ``time zones''~\cite{ernst2018sllet}.
The reactors model~\cite{Lohstroh:2019:CyPhy} in \lflong generalizes LET by
distinguishing logical time from physical time~\cite{LeeLohstroh:22:LET}.
Another emerging model offering similar generalization is the Sparse Synchronous
Model (SSM)~\cite{hui2022ssm}, the semantics of which can be facilitated using
RP2040's PIO instruction set~\cite{hui2023timestamp}, sharing similarity with
our approach.
The difference here is that it seems challenging to generate quasi-static
schedules for SSM, while we show that it is feasible for a subset of \lf.

In the real-time scheduling literature, various hardware and software
complexities have been considered in real-time systems, leading to a range of
Directed Acyclic Graph (DAG) representations. 
% These representations span from a single DAG modeling a single task~\cite{Baruah12, Graham69, he2019intra, zhao2020dag, he2021response, he2022bounding, sun2023edge} to DAGs that handle multiple tasks with different periods~\cite{Baruah14, melani2015response, Pathan18, verucchi2020latency, zhao2022dag}.
One of the earliest contributions to DAG scheduling introduced what is now known
as Graham's bound. It provides an upper bound on the response time of a task
based on the longest path within the DAG and the DAG's volume~\cite{Graham69}. 
% This bound is applicable to any work-conserving scheduling policy on
% homogeneous multicore platforms. 
% Recently, He et al.~\cite{he2022bounding} revealed the limitations of Graham's bound and proposed a tighter bound that considers multiple long paths instead of just the longest one.
% Melani et al.~\cite{melani2015response} extended Graham's bound to encompass systems with multiple DAG tasks, taking into account inter-task interference. They investigated the use of global earliest-deadline-first (EDF) and fixed-priority (FP) scheduling policies in this context.
In recent years, He et al.\cite{he2019intra} proposed a method of prioritizing subtasks within the longest paths to reduce response times and enhance schedulability. 
Subsequently, Zhao et al.\cite{zhao2020dag} improved the priority assignment strategy at the subtask level by considering subtask dependencies.
Different from the above approaches, Sun et. al~\cite{sun2023edge} proposed a
new DAG scheduler named Edge Generation Scheduling (EGS). 
% Instead of converting
% the DAG scheduling challenge into a prioritized method, they proposed adding
% edges to the original graph until the DAG becomes trivially schedulable. We
% integrated the EGS scheduler into our \pretvm-based \lf compiler due to its good performance
% in achieving high schedulability and reducing the number of workers.

\section{Conclusion}
In this paper, we present the Precision-Timed Virtual Machine, \pretvm, which
facilitates the execution of quasi-static schedules compiled from a subset of
\lf programs, supporting a top-down, model-based design flow for real-time
software in Cyber-Physical Systems.
Our approach is amenable to predicting the WCETs of \lf programs' hyperperiods
given assumptions on the WCETs of individual reactions and virtual instructions. 
We evaluate our 
approach using design patterns found in modern CPSs, implemented in
\lf. The results show that \pretvm delivers time-accurate deterministic
execution.

% Acknowledgment
\section*{Acknowledgment}
The work in this paper was supported in part by the National Science Foundation
(NSF), award \#CNS-2233769 (Consistency vs. Availability in Cyber-Physical
Systems), by DARPA grant FA8750-20-C-0156, by Intel, and by the iCyPhy Research
Center (Industrial Cyber-Physical Systems), supported by Denso, Siemens, and
Toyota. 

This work was also supported, in part, by the German Federal Ministry of
Education and Research (BMBF) as part of the program “Souverän. Digital.
Vernetzt.”, joint project 6G-life (16KISK001K), and by the German Research
Council (DFG) through the InterMCore project (505744711).

Mirco Theile and Binqi Sun were supported by the Chair for Cyber-Physical
Systems in Production Engineering at TUM. 

The authors thank Linh Thi Xuan Phan and the anonymous reviewers for
providing helpful feedback.
Shaokai Lin thanks Zitao Fang and Yang Huang for their collaboration on a course
project for CS267 (Application of Parallel Computers) at UC Berkeley, which
provided early inspiration for this work.

% Bibliography
\bibliographystyle{IEEEtran}
\bibliography{refs.bib}
\vfill

\clearpage
% \section*{Appendix}

\end{document}